\documentclass[10pt, a4paper,reqno]{amsart}

\usepackage{color} \definecolor{bleu_sombre}{rgb}{0,0,0.6}  \definecolor{rouge_sombre}{rgb}{0.8,0,0}\definecolor{vert_sombre}{rgb}{0,0.6,0}

\usepackage[plainpages=false,colorlinks,linkcolor=bleu_sombre,citecolor=rouge_sombre,urlcolor=vert_sombre,breaklinks]{hyperref}

\usepackage[T1]{fontenc}	
\usepackage[latin1]{inputenc}	
\usepackage[cm]{aeguill}	
\usepackage[english]{babel}

\usepackage{amsmath,amssymb,amsthm,amsfonts}
\usepackage{url,color}
\usepackage{nicefrac,enumerate,stmaryrd,dsfont}

\theoremstyle{plain}
\newtheorem{theorem}{{Theorem}}[section] 
\newtheorem*{theorem*}{{Theorem}}
\newtheorem{proposition}[theorem]{Proposition}
\newtheorem*{proposition*}{Proposition}
\newtheorem{corollary}[theorem]{Corollary}
\newtheorem*{corollary*}{Corollary}
\newtheorem{lemma}[theorem]{Lemma}
\newtheorem*{lemma*}{Lemma}

\theoremstyle{definition}
\newtheorem{definition}[theorem]{Définition}
\newtheorem*{definition*}{Définition}

\theoremstyle{remark}
\newtheorem*{remarque*}{Remarque}

\newtheorem{remark}[theorem]{Remark}

\newtheorem*{exemple*}{Exemple}

\newtheorem*{exemples*}{Exemples}

\newcommand{\commm}[1]{{}}

\makeatletter

\@addtoreset{equation}{section}  
\makeatother

\newcommand{\hou}{h\in]0,1]}

\renewcommand{\leq}{\leqslant}	\renewcommand{\geq}{\geqslant}
\renewcommand{\bar}[1]{\overline{#1}}
\renewcommand\over[2]{{\,\buildrel #1\over#2\,}}

\newcommand{\inv}{^{-1}}

\newcommand {\limt}[2]{\xrightarrow[#1 \to #2]{}}

\newcommand{\abs}[1]{\left\vert #1\right\vert}        
\newcommand{\nr}[1]{\left\Vert #1\right\Vert}         
\newcommand{\innp}[2]{\left< #1 , #2 \right>}         

\newcommand{\Dom}{\Dc}			
\newcommand{\Op}{{\mathop{\rm{Op}}}_h}
\newcommand{\Opw}{{\mathop{\rm{Op}}}_h^w}

\newcommand{\pppg}[1] {\left< #1 \right>} 	
\newcommand{\symbor}{C_b^\infty}		
\newcommand{\symb} {\Sc}		

\newcommand{\bigo}[2]{\mathop{O}\limits_{#1 \to #2}}
\newcommand{\littleo}[2]{\mathop{o}\limits_{#1 \to #2}}

\newcommand{\singl}[1]{\left\{ #1 \right\}}		
\newcommand{\Ii}[2]{\llbracket #1,#2 \rrbracket}	
\newcommand{\R}{\mathbb{R}}		\newcommand{\C}{\mathbb{C}}
\newcommand{\N}{\mathbb{N}}	
\newcommand{\Sph}{\mathbb{S}}
\newcommand{\1}[1]{\mathds 1 _{#1}}

\newcommand{\tqe}{\,:\,}					

\newcommand{\seq}[2]{\left({#1}_{#2}\right)_{#2 \in\N}} 

\renewcommand{\Re}{\mathop{\rm{Re}}\nolimits}        
\renewcommand{\Im}{\mathop{\rm{Im}}\nolimits}        

\DeclareMathOperator{\Hess}{Hess}			
    
\DeclareMathOperator{\Idm}{I} 
\DeclareMathOperator{\supp}{supp}                    

\newcommand{\eqv}{\Longleftrightarrow}               

\renewcommand{\a}{\alpha}\renewcommand{\b}{\beta}\newcommand{\g}{\gamma}\newcommand{\G}{\Gamma}\renewcommand{\d}{\delta}\newcommand{\D}{\Delta}\newcommand{\e}{\varepsilon}\newcommand{\z}{\zeta} \renewcommand{\th}{\theta}\renewcommand{\l}{\lambda}\renewcommand{\L}{\Lambda}\newcommand{\m}{\mu}\newcommand{\n}{\nu}\newcommand{\x}{\xi}\newcommand{\s}{\sigma}\renewcommand{\t}{\tau}\newcommand{\vf}{\phi}\newcommand{\h}{\chi}\newcommand{\p}{\psi}\renewcommand{\o}{\omega}\renewcommand{\O}{\Omega}

\newcommand{\Bc}{{\mathcal B}}\newcommand{\Dc}{{\mathcal D}}\newcommand{\Lc}{{\mathcal L}}\newcommand{\Nc}{{\mathcal N}}\newcommand{\Rc}{{\mathcal R}}\newcommand{\Sc}{{\mathcal S}}\newcommand{\Uc}{{\mathcal U}}\newcommand{\Vc}{{\mathcal V}}\newcommand{\Wc}{{\mathcal W}}\newcommand{\Zc}{{\mathcal Z}}

\newcommand{\loc}{{\rm{loc}}}

\newcommand{\hh}{H_h} \newcommand{\huh}{H_1^h}  
\newcommand{\udh}{U_2^h}   \newcommand{\tudh}{\tilde U_2^h}
\newcommand{\uh}{U_h} \newcommand{\uuh}{U_1^h}  \newcommand{\uoh}{U_0^h} 

\newcommand{\zoneS}{\mathcal Z}
\newcommand{\dpar}{\Dc _r}
\newcommand{\dparz}{\Dc _r^\z}
\newcommand{\ldgr}{{\Sph_r}}  
\newcommand{\Sphr}{{\Sph_r}}  

\newcommand{\nlpds}[1]{{L^{2,#1}(\R^n)}}
\newcommand{\nlpdsp}[2]{{L^{2,#1}(#2)}}

\newcommand{\negg}{{N_E \G}}
\newcommand{\sgamma}{{\s_\G}}
\newcommand{\snegg}{{\s_{\negg}}}

\newcommand{\hdh}{H^2_{h}}	\newcommand{\hth}{H^3_{h}}
\newcommand{\hhm}{H_{h_m}}	\newcommand{\huhm}{H^1_{h_m}}		\newcommand{\hthm}{H^3_{h_m}}
\newcommand{\Opwm}{{\mathop{\rm{Op}}}_{h_m}^w}

\begin{document}

\title{Uniform resolvent estimates for a non-dissipative Helmholtz equation}
\author{Julien Royer}


\maketitle

\begin{abstract}
We study the high frequency limit for a non-dissipative Helmholtz equation. We first prove the absence of eigenvalue on the upper half-plane and close to an energy which satisfies a weak damping assumption on trapped trajectories. Then we generalize to this setting the resolvent estimates of Robert-Tamura and prove the limiting absorption principle. We finally study the semiclassical measures of the solution when the source term concentrates on a bounded submanifold of $\R^n$.
\end{abstract}

\section{Introduction and statement of the main results}

The purpose of this paper is to study on $\R^n$, $n\geq 1$, the high frequency limit for the Helmholtz equation in a non-dissipative setting. After rescaling, this equation can be written 
\begin{equation} \label{helmholtz}
 (\hh - E)u_h = f_h, \quad \text{where} \quad  \hh = -h^2 \D +V_1(x) -ih V_2(x).
\end{equation}
We recall that this equation models for instance the propagation of the electromagnetic field of a laser in an inhomogeneus material medium. In this setting $V_1(x) - E$ is linked to the refraction index, $V_2(x)$ is the absorption index and $f_h$ is the source term. The parameter $h>0$ is proportional to the wavelength. In this paper we are interested in the asymptotic behavior of the solution $u_h$ when $h$ goes to 0. \\

The domain of $\hh$ is the Sobolev space $H^2(\R^n)$. All along this paper, we assume that $V_1$ and $V_2$ are bounded and go to 0 at infinity. This implies in particular that the essential spectrum of $\hh$ is $\R_+$ as for the free laplacian.
Our purpose is to study the resolvent $(\hh-z)\inv$, where $z \in  \C_+ = \{ \Im z >0\}$ is close to $E \in \R_+^*$. We prove some estimates for this resolvent uniform in the spectral parameter $z$, in order to obtain the limiting absorption principle and then existence and uniqueness of an outgoing solution $u_h$ for \eqref{helmholtz}. We also control the dependance in $h$ of these estimates. This gives an a priori estimate for the size of $u_h$ when $h$ goes to 0. Note that it is not clear that the resolvent is well-defined. More precisely the operator $\hh$ may have isolated eigenvalues in a strip of size $O(h)$ around the real axis. Therefore we first have to prove that it cannot happen where we study the resolvent.\\

%
%
%
%
%
%





Let $\d > \frac 12$. In the self-adjoint case ($V_2 = 0$) it is known that there exist a neighborhood $I$ of $E$, $h_0 >0$ and $c \geq 0$ such that 
\begin{equation} \label{estim-aa}
 \forall h \in ]0,h_0], \quad \sup_{\substack{\Re z \in I \\ \Im z \neq 0}} \nr { \pppg x ^{-\d} (\huh-z)\inv \pppg x ^{-\d}}_{\Lc(L^2(\R^n))} \leq \frac c h
\end{equation}
if and only if the energy $E$ is non-trapping. Here we denote by $\huh$ the self-adjoint Schrödinger operator $- h^2 \D + V_1(x)$, by $\Lc(L^2(\R^n))$ the space of bounded operators on $L^2(\R^n)$, and $\pppg x = \big( 1 + \abs x\big)^{\frac 12}$. D. Robert and H. Tamura \cite{robertt87} proved that the non-trapping condition is sufficient and X.P. Wang \cite{wang87} proved its necessity. In fact, if the non-trapping condition is not satisfied then the norm in \eqref{estim-aa} is at least of size $\abs {\ln h} / h$ (see \cite{bonybr10}).\\

For this result and all along this paper the potential $V_1$ is assumed to be of long range: it is smooth and there exist constants $\rho > 0$ and $c_\a \geq 0$ for $\a \in\N^n$ such that 
\[
 \forall \a \in\N^n, \forall x \in\R^n, \quad \abs {\partial^\a V_1(x)} \leq c_\a \pppg x^{-\rho - \abs \a}.
\]

Let $p : (x,\x) \mapsto \x^2 + V_1(x)$ be the semiclassical symbol of $H_1$ on $\R^{2n} \simeq T^* \R^n$ and $\vf^t$ the corresponding hamiltonian flow. For any $w \in \R^{2n}$, $t \mapsto \vf^t (w) = \big( \bar x (t,w) , \bar \x (t,w) \big)$ is the solution of the system
\begin{equation} \label{syst-ham} 
\left\{ \begin{array}{l} \partial _t \bar x (t,w) = 2 \bar \x(t,w) ,  \\
        \partial_ t \bar \x (t,w) = -  \nabla V_1 (\bar x(t,w)), \\
	\vf^0(w) = w.
        \end{array} \right.
\end{equation} 
We recall that $E>0$ is said to be non-trapping if 
\[
 \forall w \in p\inv(\singl E), \quad \abs{ \bar x (t,w)} \limt t {\pm \infty} + \infty.
\]
For $I \subset \R$, we introduce the following subsets of $p\inv(I)$:
\begin{eqnarray*}
\O_b^\pm(I) &=& \singl{ w \in p \inv(I) \tqe \sup_{t \geq 0} \abs{\bar x(\pm t,w)} < \infty    },\\
\O_b(I) &=& \O_b^-(I) \cap \O_b^+(I), \\
\O_\infty^\pm(I) &=& \singl{ w \in p\inv(I) \tqe  \abs{ \bar x(\pm t,w)} \limt t {+ \infty} +\infty}.
\end{eqnarray*}

In \cite{art_mourre} we considered the dissipative case $V_2 \geq 0$. We proved \eqref{estim-aa} for $\Im z >0$ under a damping assumption on trapped trajectories: 
\begin{equation} \label{am-diss}
 \forall w  \in \O_b(\singl E) , \exists T \in \R , \quad V_2\big(\bar x (T,w)\big) > 0.
\end{equation}
This generalizes the usual non-trapping condition since we recover the condition that there is no trapped trajectory when $V_2 \neq 0$. 

To prove this result we developped a dissipative version of Mourre's theory \cite{mourre81}, which we applied to the dissipative Schrödinger operator. For this we constructed an escape function as introduced by Ch. Gérard and A. Martinez \cite{gerardm88}, using the damping assumption to allow trapped trajectories. Note that L. Aloui and M. Khenissi also proved some resolvent estimates for a dissipative Schrödinger operator in \cite{alouik07}. They need a similar assumption but used a different approach (see below).\\

We know that assumption \eqref{am-diss} is both sufficient and necessary in the dissipative setting. Our purpose is now to relax the dissipative condition, allowing negative values for the absorption index $V_2$. This means that the Schrödinger operator $\hh$ is not only non-selfadjoint but also non-dissipative. 

If $V_2$ can take negative values, the damping assumption need reformulating. The condition we are going to use in this paper is the following:
\begin{equation} \label{hyp-amfaible}
\forall w \in \O_b(\singl E), \exists T > 0, \quad \int_0^T V_2( \bar x (t,w)) \, dt > 0.
\end{equation}
%
This condition is in particular satisfied when $V_2 \geq 0$ and \eqref{am-diss} holds. From this point of view, the results we are going to prove here are stronger than those given in the dissipative setting. \\

In this setting we cannot use the dissipative version of Mourre's commutators method. We use the same approach as in \cite{alouik07} instead. The idea is due to G. Lebeau \cite{lebeau96} and N. Burq \cite{burq02}. It is a contradiction argument. We consider a family of functions which denies the result, a semiclassical mesure associated to this family and finally we prove that this measure is both zero and non-zero. This idea was used in \cite{burq02} for a general self-adjoint and compactly supported perturbation of the laplacian. In \cite{jecko04}, Th. Jecko used the argument to give a new proof of \eqref{estim-aa} with a real-valued potential. The motivation was to give a proof which could be applied to matrix-valued operators. To allow long range potentials, the author introduced a bounded ``escape  function'' which we use here. The method was then used in \cite{castellajk08} for a potential with Coulomb singularities and in \cite{jecko05,fermanianr08,duyckaertsfj09} for a matrix-valued operator.\\

Let us now state the main results about the resolvent. An important difference with the dissipative case is that we do not know if the resolvent is well-defined, even on the upper half-plane $\C_+$. So in the results we state now, we first claim that $\hh$ has no eigenvalue in the considered region and then give an estimate for the resolvent. The first theorem is about spectral parameters whose imaginary parts are bigger than $\b h $ for some $\b > 0$:

\begin{theorem} \label{th-vp}
Suppose $V_2$ is smooth with bounded derivatives and $V_2 (x) \to 0$ when $\abs x \to +\infty$. Let $E > 0$ be an energy which satisfies the damping assumption \eqref{hyp-amfaible} and $\b > 0$. Then there exist a neighborhood $I$ of $E$, $h_0 > 0$ and $c \geq 0$ such that for $h \in ]0,h_0]$ and
\[
z \in \C_{I,h\b} = \singl{ z\in\C \tqe \Re z \in I, \Im z \geq h \b}
\]
the operator $(\hh-z)$ has a bounded inverse and
\[
\nr{(\hh-z)\inv}_{\Lc(L^2(\R^n))} \leq \frac c {h}.
\]
\end{theorem}

Note that this result is obvious when $\hh$ is self-adjoint or at least dissipative. We can take $c = \b \inv$ in these cases.\\

In a second step we study the resolvent up to the real axis, generalizing \eqref{estim-aa}. Now $V_2$ has to be of short range, which means that there exist $\rho > 0$ and constants $c_\a \geq 0$  for $\a \in \N$ such that
\[
\forall \a \in \N^n, \forall x \in\R^n, \quad  \abs{\partial^\a V_2 (x)} \leq c_\a \pppg x^{-1-\rho-\abs \a}.
\]

\begin{theorem} \label{th-estim}
Assume that $V_2$ is of short range. Let $E > 0$ satisfy the damping assumption \eqref{hyp-amfaible} and $\d > \frac 12$. Then there exist a neighborhood $I$ of $E$, $h_0 > 0$ and $c \geq 0$ such that for $h\in ]0,h_0]$ and
\[
z \in \C_{I,+} = \singl{ z\in\C \tqe \Re z \in I, \Im z >0}
\]
the operator $(\hh-z)$ has a bounded inverse and
\[
\nr{\pppg x ^{-\d} (\hh-z)\inv \pppg x ^{-\d}}_{\Lc(L^2(\R^n))} \leq \frac c h.
\]
\end{theorem}

The proof of this theorem is inspired from \cite{jecko04}. In particular we use a bounded escape function at infinity to prove that the semiclassical measure we study is non-zero. But contrary to the selfadjoint case (this could also be done in the dissipative case) we cannot use this escape function to prove that this measure is supported in a compact subset of $\R^{2n}$. We use instead the estimate for the outgoing solution to the Helmholtz equation in the incoming region proved in \cite{robertt89}.\\

As in the dissipative case we have only given a result on the upper half-plane. Here we have no assumption about the sign of $V_2$, but there still is a damping condition in \eqref{hyp-amfaible}. The difference with the dissipative context is that we recover a symmetric situation under the stronger non-trapping condition, so that the result we have proved for $\Im  z >0$ now holds for $\Im z < 0$:

\begin{corollary} \label{cor-estim}
Assume that $V_2$ is of short range. Let $E > 0$ be a non-trapping energy and $\d > \frac 12$. Then there exist a neighborhood $I$ of $E$, $h_0 >0$ and $c \geq 0$ such that for $h \in ]0,h_0]$, $\Re z \in I$ and $\Im z \neq 0$ the operator $(\hh-z)$ has a bounded inverse and 
\[
\nr{\pppg x ^{-\d} (\hh-z)\inv \pppg x ^{-\d}}_{\Lc(L^2(\R^n))} \leq \frac c h.
\]
\end{corollary}

Once we have the uniform resolvent estimates, we can prove the limiting absorption principle. This question has been studied for a long range self-adjoint Schrödinger operator in \cite{ikebes72} and \cite{saito}. It is proved that the equation $(H-E)u = f$ has a unique outgoing solution $u \in H^2_{\loc}(\R^n) \cap L^{2,-\d}(\R^n)$ when $f \in L^{2,\d}(\R^n)$, and this solution is given as the limit in $L^{2,-\d}(\R^n)$ of $(H-z)\inv f$ when $z \in \C_+$ goes to $E$ (we do not work in the semiclassical limit here, so we only consider the case $h=1$). An outgoing solution is a solution which satisfies a radiation condition of Sommerfeld type at infinity (see Definition \ref{def-outgoing}). The strategy is to prove first uniqueness of an outgoing solution. This is used to prove resolvent estimates, and then we can get the limiting absorption principle. The result has been extended to the non-selfadjoint case in \cite{saito74}. Y. Saito proves that when the potential has a short range imaginary part, the result can be extended where we have uniqueness of the ougoing solution, which no longer holds for any $E >0$. Here we consider the case where the dissipative part has a long range positive part and a short range negative part. Contrary to these papers, we use the fact that we already have uniform resolvent estimates to obtain uniqueness of the outgoing solution and then the limiting absorption principle. With Theorem \ref{th-estim} we obtain the following result:

\begin{theorem} \label{th-lap}
Let $V_2$, $E$ and $\d$ be as in Theorem \ref{th-estim}. Then there exist a neighborhood $I$ of $E$ and $h_0 >0$ such that for any $h \in ]0,h_0]$, $\l \in I$ and $f \in L^{2,\d}(\R^n)$ the limit
\[
 \lim_{\substack{z \to \l \\ z \in \C_{I,+}}} (\hh-z)\inv f
\]
exists in $L^{2,-\d}(\R^n)$ and defines the unique outgoing solution for the equation $(\hh -\l)u = f$ (see Definition \ref{def-outgoing}).
\end{theorem}

Note that when the dissipative part $V_2$ is non-negative we can prove directly uniqueness of the outgoing solution and hence we can proceed as in the self-adjoint case (see Proposition \ref{prop-lap-diss}).\\

Now that the outgoing solution 
\[
 u_h = (\hh -(E+i0))\inv f_h
\]
is well-defined in $L^{2,-\d}(\R^n)$ for any $h\in ]0,h_0]$ ($h_0 >0$ being given by Theorem \ref{th-lap}) and $f_h \in L^{2,\d}(\R^n)$, we can study its semiclassical measures when the source term $f_h$ concentrates on a submanifold of $\R^n$. We recall that a measure $\m$ on the phase space $T^*\R^n \simeq \R^{2n}$ is said to be a semiclassical measure for the family $(u_h)_{ h \in ]0,h_0]}$ if there exists a sequence $\seq h m \in ]0,h_0]^\N$ such that
\[
h_m \limt m \infty 0 \quad \text{and} \quad  \forall q \in C_0^\infty(\R^{2n}), \quad \innp{\Opwm u_{h_m}}{u_{h_m}} \limt m \infty \int_{\R^{2n}} q \, d\m.
\]
Here $\Opw(q)$ denotes the Weyl $h$-quantization of the symbol $q$:
\[
 \Opw(q) u (x) = \frac 1{(2\pi h)^n} \int_{\R^n} \int_{\R^n} e^{\frac ih \innp{x-y} \x} q \Big( \frac {x+y} 2 , \x \Big) u(y) \, dy \, d\x.
\]
We will also use the standard quantization:
\[
 \Op(q) u (x) = \frac 1{(2\pi h)^n} \int_{\R^n} \int_{\R^n} e^{\frac ih \innp{x-y} \x} q (x,\x) u(y) \, dy \, d\x.
\]

The first paper in this direction is \cite {benamou-al-02}, where $f_h$ is assumed to concentrate on $\G = \singl 0$ as $h$ goes to 0 (see also \cite{castella05}). This was extended in \cite{castellapr02,wangz06} to the case where $f_h$ concentrates on an affine subspace $\G$ in $\R^n$. J.-F. Bony gives in \cite{bony09} another proof for the case $\G = \singl 0$ using different assumptions and above all a different approach. We used this point of view in \cite{art_mesure} to deal with the case where the absorption index is non-constant and $\G$ is any bounded submanifold of dimension $d \in \Ii 0 {n-1}$ in $\R^n$. Trapped trajectories of energy $E$ for the classical flow were allowed under assumption \eqref{am-diss}. Our purpose here is to check that this result still holds --under hypothesis \eqref{hyp-amfaible}-- when $V_2$ takes negative values. The proof is actually approximately the same. We will have to prove the estimates in the incoming region in this non-dissipative setting (see Theorem \ref{th-incoming}) and be careful with the fact that the semi-group generated by $\hh$ is no longer a contraction semi-group.\\

Throughout this paper we denote by $C_0^\infty(\R^n)$ the set of smooth and compactly supported functions on $\R^n$, by $\symbor(\R^n)$ the set of smooth functions whose derivatives are bounded, and $\Sc(\R^n)$ is the Schwartz space. For $\d \in \R$, we denote by $\Sc\big(\pppg x ^\d \big)$ the set of symbols $a \in C^\infty(\R^{2n})$ such that
\[
 \forall \a,\b\in\N^n, \exists c_{\a,\b} \geq 0, \forall (x,\x) \in \R^{2n} , \quad \abs{\partial_x^\a \partial_\x^\b a(x,\x)} \leq c _{\a,\b} \pppg x^{\d}.
\]
We also denote by $\symb_{\d}(\R^{2n})$ the set of symbols $a \in C^\infty(\R^{2n})$ such that
\[
 \forall \a,\b\in\N^n, \exists c_{\a,\b} \geq 0, \forall (x,\x) \in \R^{2n} , \quad \abs{\partial_x^\a \partial_\x^\b a(x,\x)} \leq c _{\a,\b} \pppg x^{\d - \abs \a}.
\]
For $R > 0$ we denote by $B_R$ the open ball of radius $R$ in $\R^n$, by $B_R^c$ its complement in $\R^n$, and by $\Sph_R$ the sphere of radius $R$, endowed with the Lebesgue measure. We also set $B_x(R) = \singl{(x,\x)  \tqe \abs x <R} \subset \R^{2n}$.\\

In Section 2 we recall some properties the flow $\vf^t$ defined by \eqref{syst-ham} and discuss assumption \eqref{hyp-amfaible}. Section 3 is devoted to the proof of Theorem \ref{th-vp}. Before giving a proof of Theorem \ref{th-estim} in Section 5, we state a non-dissipative version for the estimate of the outgoing solution for the Helmholtz equation in the incoming region (see Section 4). With the uniform resolvent estimates, we prove the limiting absorption principle in Section 6. We finally show in Section 7 that the result known for the semiclassical measure for the outgoing solution of \eqref{helmholtz} when the source term concentrates on a bounded submanifold of $\R^n$ remains valid in our non-dissipative setting.
\\

\subsection*{Acknowledgement} This work is partially supported by the French National Research Project NONAa, No. ANR-08-BLAN-0228-01, entitled {\em Spectral and microlocal analysis of non selfadjoint operators}.

\section{More about classical dynamics and the weak damping assumption}

The good properties of the flow at infinity come from the fact that for any $\nu > 0$ there exists $\Rc \geq 0$ such that
\begin{equation} \label{escape-rad}
 \forall x \in \R^n , \quad \abs x \geq \Rc \implies \abs{V_1(x)}  + \abs x \abs {\nabla V_1(x)} < \nu.
\end{equation}
This means that far from the origin the refraction index has low influence on the flow and hence the classical trajectories behaves as in the free case.\\

In particular for $J = [E_1,E_2]\subset \R_+^*$ and $\Rc$ such that \eqref{escape-rad} holds for $\nu = \frac {2E_1}3$, we have $\partial_t^2 \abs {\bar x(t,w)}^2 \geq 8E_1 - 12\nu >0$ if $p(w) \in J$ and $\abs {\bar x(t,w)} \geq \Rc$. As a consequence a classical trajectory of energy $E \in J$ which leaves $B_x(\Rc)$ cannot come back and goes to infinity. This implies that
\[
p\inv (J) = \O_b^+(J) \sqcup \O_\infty^+(J) = \O_b^-(J) \sqcup \O_\infty^-(J) =  \O_b(J) \cup \O_\infty^+(J) \cup \O_ \infty^-(J).
\]
Moreover $\O_b^\pm(J)$ is closed in $\R^{2n}$ and $\O_b(J) \subset B_x(\Rc)$ is compact. If $\Bc_\pm$ is a bounded subset of $\O_b^\pm(J)$, we choose $\Rc$ such that \eqref{escape-rad} holds for $\nu = \frac {2E_1}3$ and $\Bc_\pm \subset B_x(\Rc)$ to prove that the set
\[
 \singl{\vf^{\pm t}(w), t \geq 0, w \in \Bc_\pm} 
\]
is bounded in $\R^{2n}$.
\\

For $R \geq 0$, $d\geq 0$ and $\s \in [-1,1]$ we denote by
\[
\zoneS_\pm(R,d,\s)  = \singl{ (x,\x)\in\R^{2n} \tqe \abs x \geq R,   \abs \x  \geq d  \text{ and } \innp x \x \gtreqless\s \abs x \abs \x} 
\]
the incoming and outgoing regions. The proposition we prove now ensures that a trajectory starting outside some incoming region and far enough from the origin stays away from the influence of $V_1$:

\begin{proposition} \label{prop-non-entr}
Let $E_2 \geq E_1 > 0$, $J \subset [E_1,E_2]$ and $\s\in [0,1[$ such that $\s^2 E_2 < E_1$. Then there exist $\Rc > 0$ and $c_0$ such that
\[
 \forall t \geq 0, \forall (x,\x) \in \zoneS_\pm (\Rc,0,\mp \s) \cap p\inv(J), \quad \abs{\bar x(\pm t,x,\x)} \geq c_0 (t + \abs x).
\]
\end{proposition}

\begin{proof}
Let $\tilde \Rc$ such that \eqref{escape-rad} holds for $\nu \in \left]0, \frac {2E_1}3\right[$ so small that
\[
\tilde \nu = 1 - \s^2 \frac {E_2 + \nu}{E_1 - \frac 32 \nu} > 0.
\]
Let $\Rc$ be greater than ${\tilde \Rc}/\sqrt {\tilde \nu}$ and $(x,\x) \in \zoneS_\pm (\Rc,0,\mp \s) \cap p\inv(J)$. Suppose there exists $t\geq 0 $ such that $\abs{\bar x (\pm t,x,\x)} < \tilde \Rc$ and let
\[
t_0 = \inf \singl{t\geq 0, \abs{\bar x (\pm t,x,\x)} < \tilde \Rc}.
\]
Let $E_3 = E_1 - \frac 32 \nu > 0$. For $t \in [0,t_0]$ we have $\partial_t^2 \abs{x(\pm t,x,\x)}^2\geq 8 E_3$ and hence
\begin{equation*} 
\abs{x(\pm t,x,\x)}^2 \geq \abs{x}^2  - 4t\s \abs x \abs\x + 4  E_3 t^2\geq \abs x^2 \left( 1 -  \frac {\s^2  \abs \x^2}{E_3}\right) \geq \abs x ^2 \tilde \nu > \tilde \Rc ^2.
\end{equation*}
This gives a contradiction when $t = t_0$. This proves that these inequalities actually hold for all $t \geq 0$. We also have
\[
\forall t \geq  {2 \s \abs x \abs \x} E_3, \quad \abs {\bar x (\pm t,x,\x)}^2 - \abs x^2 - 2E_3 t ^2 \geq 2t (E_3 t - 2\s \abs x \abs \x) \geq 0,
\]
which concludes the proof.
\end{proof}

We now discuss assumption \eqref{hyp-amfaible}. We still denote by $V_2$ the function $(x,\x) \mapsto V_2(x)$ on $\R^{2n}$. We first remark that since $(x,\x)$ and $(x,-\x)$ are simultaneously in $\O_b(\singl E)$, assumption \eqref{hyp-amfaible} could be equivalently formulated looking at trajectories in the past:
\begin{equation} \label{am-diss-past}
\eqref{hyp-amfaible} \quad  \eqv \quad  \forall w  \in \O_b(\singl E) , \exists T \in \R , \quad V_2\big(\bar x (-T,w)\big) > 0.
\end{equation}
Using compactness of $\O_b([E/2,2E])$, we can also check that assumption \eqref{hyp-amfaible} is an open property:

\begin{proposition}
If assumption \eqref{hyp-amfaible} holds for some $E > 0$, then it also holds for any $\l$ in some neighborhood of $E$.
\end{proposition}

%

We now prove that assumption \eqref{hyp-amfaible} ensures that the absorption is positive on trapped trajectories ``on average'' in time.

\begin{proposition}  \label{prop-gros-amort}
Let $J \subset \R_+^*$ be such that assumption \eqref{hyp-amfaible} holds for any $\l\in J$. Then for all compact $K \subset \O_b^\pm(J)$ there exist $c_0>0$ and $C \geq 0$ such that 
\[
\forall t \geq 0, \forall w \in K, \quad  \int_0^t (V_2 \circ \vf^{\pm s} )(w)\, ds  \geq  c_0 t - C.
\]
\end{proposition}

\begin{remark}
 This proposition implies that 
\[
 \forall w \in \O_b^\pm(J), \quad \liminf_{t\to +\infty} \frac 1  t  \int_0^t (V_2 \circ \vf^{\pm s} )(w)\, ds > 0.
\]
The claim that the average in time of the absorption is positive on trapped trajectories becomes clear on periodic trajectories. If $w \in p\inv(J)$ and $T > 0$ are such that $\vf^T(w) = w$, then 
\[
  \int_0^T (V_2 \circ \vf^{t} )(w)\, dt > 0.
\]
To prove this we only have to apply \eqref{hyp-amfaible} to $w_0 = \vf^{t_0}(w) \in \O_b(\singl E)$, where $t_0 \in [0, T]$ is the time for which $t \mapsto \int_0^t V_2(\bar x(s,w))\, ds$ reaches its maximum.
\end{remark}

\begin{proof}[Proof of Proposition \ref{prop-gros-amort}]
Since $K$ is compact we can asssume without loss of generality that $J$ is a compact subset of $\R_+^*$.\\

\noindent
{\bf 1.}
Let $w \in \O_b(J)$. By assumption there exist $T_w,\g_w > 0$ such that 
\[
\int_0^{T_w} (V_2 \circ \vf^{\pm s} )(w)\, ds  \geq 2 \g_w.
\]
Since the left-hand side is a continuous function of $w$, we can find a neighborhood $\Vc_w$ of $w$ in $\R^{2n}$ such that for all $v \in \Vc_w$ we have
\[
\int_0^{T_w} (V_2 \circ \vf^{\pm s} )(v)\, ds \geq \g_w.
\]
As $\O_b(J)$ is compact, it is covered by a finite number of such sets $\Vc_w$. Hence we can find $T_1 ,\g_1 > 0$ such that
\[
\forall w \in \O_b(J), \exists t \in [0,T_1], \quad \int_0^{t} (V_2 \circ \vf^{\pm s} )(w)\, ds \geq  \g_1.
\]

\noindent
{\bf 2.}
Let $\n =   \frac {2 + T_1 \nr {V_2}_\infty} {\g_1}$ and $T_2 := T_1 \left( 1 +  \nu \right)$. Let $w \in \O_b(J)$. We set $t_0 = 0$ and for all $k \in \N$ we consider by induction $t_{k+1} \in ]t_k,t_k + T_1]$ such that
\[
\int_{t_k}^{t_{k+1}}(V_2 \circ \vf^{\pm s})(w) \, ds \geq \g_1.
\]
We necessarily have $t_{k+1} \geq t_k  + \g_1 / \nr {V_2}_\infty$ for all $k \in \N$, and hence $t_k \to +\infty$ (if $V_2=0$ the statement of the proposition is empty). In particular any $t > 0$ belongs to $]t_k,t_k+T_1]$ for some $k\in\N$. Let $t \geq T_2$ and $N \in \N$ such that $t \in ]t_N,t_N+T_1]$. We have $N\geq \nu$ and hence
\[
\begin{aligned}
\int_0^{t} (V_2 \circ \vf^{\pm s} )(w)\, ds \geq \sum_{k=0}^{N-1} \int_{t_k}^{t_{k+1}} (V_2 \circ \vf^{\pm s} )(w)\, ds + \int_ {t_N}^t (V_2 \circ \vf^{\pm s} )(w)\, ds \geq N \g_1 - T_1 m_-,
\end{aligned}
\]
where $m_- = - \inf V_2$. This proves that
\[
\forall w \in \O_b(J), \forall t \geq T_2,  \quad  \int_0^{t} (V_2 \circ \vf^{\pm s} )(w)\, ds \geq 2 .
\]

\noindent
{\bf 3.}
By continuity, there exists a neighborhood $\Uc$ of $\O_b(J)$ such that
\[
\forall w \in \Uc, \quad \int_0^{T_2} (V_2 \circ \vf^{\pm s} )(w)\, ds \geq  1.
\]

\noindent
{\bf 4.}
Now let $K$ be a compact subset of $\O_b^\pm (J)$. 
We prove by contradiction that there exists $T_K \geq 0$ such that $\vf^{\pm t}(w) \in \Uc$ for all $w\in K$ and $t \geq T_K$. If it is not the case we can find  sequences $\seq t m$ and $\seq w m$ with $w_m \in K$ and $t_m \to +\infty$ such that $\vf^{\pm t_m}(w_m) \notin \Uc$. Each term of the sequence $(\vf^{\pm t_m}(w))_{m\in\N}$ belongs to the bounded set $\bigcup _{t\geq 0} \vf^{\pm t} (K)$, so after extracting a subsequence if necessary, we can assume that it converges to
\[
w_\infty \in \O_b^\pm (J) \setminus\Uc  \subset \O_\infty^\mp(\R).
\] 
Let $\Rc$ be such that \eqref{escape-rad} holds for $\nu = {2\inf J}/3 >0$ and $K \cup \singl {w_\infty}  \subset B_x(\Rc)$. There exists $T_\infty \geq 0$ such that $\abs{\bar x (\mp T_\infty,w_\infty)} > 2\Rc$. By continuity and properties of $\Rc$, we can find a neighborhood $\Vc \subset B_x(\Rc)$ of $w_\infty$ such that $\abs {\bar x (\mp t,v)} > 2\Rc$ for all $v \in \Vc$ and $t \geq T_\infty$. Hence for large $m$ we have $\vf^{\mp t_m}(\Vc) \cap K = \emptyset$, and in particular $\vf^{\pm t_m} (w_m) \notin \Vc$. This gives a contradiction.\\

\noindent
{\bf 5.}
Let $w \in K$, $t\geq T_K$ and $N$ be the integer part of $\frac {t-T_K}{T_2}$. We have
\begin{eqnarray*}
\lefteqn{\int_{0}^{t} (V_2 \circ \vf^{\pm s} )(w)\, ds}\\
&& \geq \int_{0}^{T_K} (V_2 \circ \vf^{\pm s} )(w)\, ds + \sum_{k=0}^{N-1} \int_{T_K+kT_2}^{T_K + (k+1) T_2} (V_2 \circ \vf^{\pm s} )(w)\, ds + \int_{T_K + N T_2}^{t} (V_2 \circ \vf^{\pm s} )(w)\, ds\\
&& \geq -T_K m_-  + N - T_2 m_-\\
&& \geq -T_K m_- + \frac {t-T_K}{T_2}-1 - T_2 m_-.
\end{eqnarray*}
Since this integral is not less than $-T_K m_-$ when $t \in [0,T_K]$, this gives the result with $c_0 = 1/T_2$ and $C = 1 + (T_K + T_2) m_- + T/T_2$.
\end{proof}

We are going to use in section \ref{sec-mesure} a more precise result:

\begin{proposition}  \label{prop-gros-amort-2}
Let $\Rc > 0$ and $J \subset \R_+^*$ such that assumption \eqref{hyp-amfaible} holds for any $\l \in J$. Then for any compact subset $\tilde K$ of $p\inv(J)$ there exist $c_0,C > 0$ such that
\[
\forall t \geq 0, \forall w \in \tilde K, \quad  \int_0^t (V_2 \circ \vf^{\pm s} )(w)\, ds  \geq  c_0 t - C \quad \text{or} \quad \abs{\bar x(\pm t ,w)} \geq \Rc.
\]
\end{proposition}

If $K \subset\O_b^\pm (J)$ this comes from Proposition \ref{prop-gros-amort}, and if $K \subset \O_\infty^\pm(J)$, the second conclusion holds for $t$ large enough, uniformly in $w \in K$, and the first conclusion is always true for finite times. The problem therefore comes from the boundary between $\O_b^\pm(J)$ and $\O_\infty^\pm(J)$.

\begin{proof}
As above we may assume that $J$ is compact. Since the conclusion is stronger if $\Rc$ is taken larger, we may assume that \eqref{escape-rad} holds for $\nu = 2\inf J /3$.
Let $K = \tilde K \cap \O_b^\pm(J)$. $K$ is a compact subset of $\O_b^\pm (J)$. We use the notation introduced in the proof of Proposition \ref{prop-gros-amort}. We know that there exists $T_K \geq 0$ such that $\vf^{\pm t}(w) \in \Uc$ for all $t \geq T_K$ and $w \in K$. By continuity of the hamiltonian flow, there exists a neighborhood $\Vc$ of $K$ in $\tilde K$ such that $\vf^{\pm T_K}(w) \in \Uc$ for all $w \in \Vc$.

We now prove that there exists $T_f \geq 0$ such that for all $w \in \Vc \setminus K \subset \O_\infty^\pm (J)$ we can find $\t_w \geq T_K$ which satisfies: 
\begin{equation} \label{tau-m}
\forall t \in [T_K,\t_w], \quad \vf^{\pm t}(w) \in \Uc \quad \text{and} \quad \forall t \geq \t_w + T_f, \quad \abs{\bar x (\pm t , w)} \geq \Rc.
\end{equation}
This means that even if we cannot say when a trajectory coming from $\Vc \setminus K$ will leave $B_x(\Rc)$, we control the time it can stay in $B_x(\Rc) \setminus \Uc$.

Assume that we cannot find $T_f$ such that \eqref{tau-m} holds. Then we can find a sequence $\seq w m$ of elements in $\Vc \setminus K$ and times $t_m\geq T_K$, $\th_m \geq m$ for $m\in\N$ such that
\[
\vf^{\pm t_m}(w_m) \notin \Uc \quad \text{and}\quad \forall t \in [t_m,t_m + \th_m],\quad  \abs{\bar x (\pm t ,w_m)} \leq \Rc.
\]

After extracting a subsequence if necessary, we may assume that $w_m$ converges to $w_\infty \in \tilde K$. If $w_\infty \in \O_\infty^\pm (J)$ then there exist $T_\infty \geq 0$ and a neighborhood $\Wc$ of $w_\infty$ such that $\abs{\bar x (\pm t,v)}\geq \Rc$ for all $t \geq T_\infty$ and $v\in \Wc$, which is impossible. This means that the limit $w_\infty$ actually belongs to $K$.

If the sequence $\seq t m$ is bounded we can assume, after extraction, that $t_m \to t_\infty \geq T_K$, which cannot be true since we would have $\vf^{\pm t_m}(w_m) \to \vf^{\pm t} (w_ \infty) \in\Uc$ and hence $\vf^{\pm t_m}(w_m) \in \Uc$ for $m$ large enough. Extracting again a subsequence, we can assume that $t_m \to +\infty$.
 
Let $v_m = \vf^{\pm t_m}(w_m)$. The sequence $\seq v m$ is bounded so without loss of generality we can assume that it converges to some $v_\infty \in p\inv(J)$. Since $t_m,\th_m \to \infty$ and the sequences $(\vf^{\mp t_m}(v_m))_{m\in\N}$ and $(\vf^{\pm \th_m}(v_m))_{m\in\N}$ are bounded, we obtain as before that $v_\infty \in \O_b(J)$, which gives a contradiction and hence proves \eqref{tau-m}.

%
The complement $\tilde K \setminus \Vc$ is a compact subset of $\O_\infty^\pm (J)$. Choosing $T_f$ larger if necessary, we can assume that 
\[
\forall w \in \tilde K \setminus \Vc, \forall t \geq T_f, \quad \abs{\bar x (\pm t ,w)} \geq \Rc.
\]
As a consequence, given $t \geq 0$ and $w \in \tilde K$ such that $\abs{\bar x (\pm t,w)} \leq \Rc$, we have $\vf^{\pm s}(w) \in \Uc$ for all $s \in [T_K,t-T_f]$ and hence
\[
 \int_{0}^{t} (V_2 \circ \vf^{\pm s} )(w)\, ds \geq -m_- T_K + \frac {t-T_K-T_f}{T_2} - 1 - T_2 m_- - T_f m_-.
\]
\end{proof}

\section{Resolvent at distance of order $h$ from the real axis}

In this section we give some general properties about the semiclassical measures we consider and prove Theorem \ref{th-vp}.

\begin{proposition} \label{prop-burq-debut}
Let $\seq z m \in \C^\N$ and $\seq h m \in ]0,1]^\N$ be sequences such that
\[
h_m \limt m \infty 0, \quad  \l_m := \Re z_m \limt m \infty E > 0 \quad \text{and} \quad \b_m := h_m\inv \Im z_m \limt m \infty \b \in \R.
\]
Consider $\d \geq 0$ and a sequence $\seq v m \in H^2(\R^n)^\N$ such that
\[
\nr{v_m}_{L^{2,-\d}(\R^n)} = 1, \quad \nr {(\hhm - z_m) v_m }_{L^{2,\d}(\R^n)} = \littleo m \infty (h_m),
\]
and
\begin{equation} \label{conv-mesure}
\forall q\in C_0^\infty(\R^{2n}), \quad \innp{\Opwm(q) v_m}{v_m}_{L^2(\R^n)} \limt m \infty \int _{\R^{2n}} q \, d\m
\end{equation}
for some (non-negative) measure $\m$ on $\R^{2n}$. Then we have the following two properties.
\begin{enumerate} [(i)]
\item If $q \in \symb \big(\pppg x ^{-2\d}\big)$ is supported outside $p\inv(J)$ for some neighborhood $J$ of $E$ we have
\[
 \innp{\Opwm (q) v_m}{v_m} \limt m \infty 0.
\]
In particular $\m$ is supported on $p\inv(\singl E)$ and for $\h \in C_ 0^\infty(\R^n)$ we have
\[
\innp{\h v_m}{v_m}_{L^2(\R^n)} \limt m \infty \int _{\R^{2n}} \h(x) \, d\m(x,\x).  
\]

\item For $q \in C_0^\infty(\R^{2n})$ and $t \geq  0$ we have
\begin{equation} \label{propagation-mu}
\begin{aligned}
\int_{\R^{2n}} q \, d\m
 = \int _{\R^{2n}} (q\circ \vf^t) \exp \left( -2\int_0^t (V_2+\b)\circ \vf^{s}\, ds\right) \, d\m.
\end{aligned}
\end{equation}
\end{enumerate}
\end{proposition}

\begin{proof} 
{\bf (i)}
For $m$ large enough (such that $\l_m \in \mathring J$) we set:
\[
a_m(x,\x) = \frac {q(x,\x) \pppg x^{2\d}}{p(x,\x)-z_m}.
\]
Since $q$ vanishes on $p\inv(J)$, we have $a_m \in \symb\big(\pppg \x^{-2}\big)$ uniformly in $m$. We can write
\[
\begin{aligned}
\innp{\Opwm(q) v_m}{v_m}
& \leq \nr{\Opwm(q)v_m}_{L^{2,\d}(\R^n)}\\
& \leq \nr{\Opwm(a_m) (\huhm-z)v_m }_{L^{2,-\d}(\R^n)}  + \littleo m \infty(h_m)\\
& \leq \nr{\Opwm(a_m) (\hhm-z)v_m }_{L^{2,-\d}(\R^n)}  + \littleo m \infty(h_m)\\
& \limt m \infty 0,
\end{aligned}
\]
which proves the first assertion. Applied with $q \in C_0^\infty(\R^{2n})$, this proves that $\m$ is supported on $p\inv(\singl E)$.
Now let $\h \in C_0^\infty(\R^n)$ and $\tilde \h \in C_0^\infty(\R^n)$ such that $\R^n \times \supp (1-\tilde \h)$ does not intersect $p\inv (\singl E)$. Then
\[
\innp{\Opwm\big(\h(x)(1-\tilde \h(\x))\big) v_m}{v_m} \limt m \infty 0,
\]
and hence
\begin{align*}
\lim_{m\to\infty} \innp{ \h(x) v_m}{v_m} 
&= \lim_{m\to \infty}\innp{ \Opw(\h(x)\tilde \h (\x)) v_m}{v_m}  =  \int_{\R^{2n}} \h(x)\tilde \h(\x) \, d\m(x,\x)\\
& = \int_{\R^{2n}} \h(x) \, d\m(x,\x).
\end{align*}

\noindent
{\bf (ii)}
Let $q \in C_0^\infty(\R^{2n})$ and $t \geq 0$. For $\t \in [0,t]$ we set
\[
q(\t,w) = q\big(\vf^{t-\t}(w)\big) \exp\left( -2 \int_\t^t (V_2+\b) \big(\vf^{s-\t} (w)\big)\,ds\right) .
\]
Since the union of $\supp (q \circ \vf^{t-\t})$ for $\t \in [0,t]$ is bounded in $\R^{2n}$ we can use differentiation under the integral sign:
\[
\begin{aligned}
\frac d {d\t} \int _{\R^{2n}} q(\t)\,d\m
& = \int _{\R^{2n}}\frac d {d \t}  q(\t)\,d\m \\
& = \int _{\R^{2n}} \big(2(V_2 +\b) q(\t) - \{ p, q(\t)\}\big)\,d\m\\
& = \lim_{m\to \infty}  \innp{ \Opwm\big( 2(V_2 +\b) q(\t) - \{ p, q(\t)\}  \big) v_m}{v_m}\\
& = \lim_{m\to \infty}  \innp{ \left(2 (V_2 + \b_m) \Opwm(q(\t)) - \frac i {h_m} \big[\huhm,\Opwm(q(\t))\big] \right) v_m}{v_m}\\
& = \lim_{m\to \infty} \frac i {h_m}  \innp{ \left(\Opwm(q(\t)) (\hhm-z_m) - (\hhm^*-\bar {z_m}) \Opwm(q(\t))   \right) v_m}{v_m}\\
& = 0.
\end{aligned}
\]
This gives statement (ii). Here we do not have to worry  about decay properties of $v_m$ since we only work with compactly supported symbols.
\end{proof}

We now turn to the proof of Theorem \ref{th-vp}. We first remark that it is easy when $\b > m_-$ since
\[
 \hh-z = \big(\hh  - ihm_- \big) - \big(z-ih m_- \big)
\]
and $\hh  - ih m_- $ is maximal dissipative. This proves that if $\Im z > h m_-$ the resolvent $(\hh-z)\inv$ is well-defined and
\begin{equation} \label{res-non-diss}
\nr{(\hh-z)\inv}_{\Lc(L^2(\R^n))} \leq \frac 1 {\Im z - hm_-}.
\end{equation}

As said in introduction we proceed by contradiction to prove the general case. So we assume we can find sequences $\seq v m \in H^2(\R^n)^\N$, $\seq z m \in \C^\N$ and $\seq h m \in ]0, 1]^\N$ such that
\[
h_m \to 0, \quad  \l_m := \Re z_m \to E, \quad  \b_m := h_m\inv \Im z_m \geq \b,  \quad \nr {v_m}_{L^2(\R^n)} = 1
\]
and
\[
 \nr{(\hhm -z_m) v_m}_{L^2(\R^n)} = \littleo m \infty (h_m).
\]

Since we already have the result for large $\b$, the sequence $\seq \b m$ is necessarily bounded. After extracting a subsequence if necessary, we can assume that $\b_m \to \tilde \b \geq \b$.
Since a bounded sequence in $L^2(\R^n)$ always has a semiclassical measure (see \cite{burq97,evansz}), we can assume after extracting another subsequence that \eqref{conv-mesure} holds for some nonnegative Radon measure $\m$ on $\R^n$. Our purpose is now to prove that $\m$ is both zero and non-zero to get a contradiction.

\begin{proposition} \label{prop-vp-nonzero}
The measure $\m$ is non-zero.
\end{proposition}

\begin{proof}
As $V_2$ goes to 0 at infinity, there exists $R \geq 0$ such that $V_2(x) \geq - \frac \b 2$ for all $x \in B_R^c$. We have
\[
\int _{\R^n} (V_2(x) + \b_m)  \abs{v_m(x)}^2\,dx = -h_m \inv  \Im  \innp{(\hhm-z_m) v_m}{v_m}_{L^2(\R^n)} \limt m \infty 0
\]
and hence, for $\h \in C_0^\infty(\R^n,[0,1])$ equal to 1 on $B_R$,
\begin{align*}
\frac \b 2
& =  \frac \b 2 \int_{\R^n} (1-\h(x)) \abs{v_m(x)}^2\,dx+ \frac \b 2 \int_{\R^n} \h(x) \abs{v_m(x)}^2\,dx \\
& \leq \int_{\R^n} (V_2(x) + \b_m) (1-\h(x)) \abs{v_m(x)}^2\,dx+ \frac \b 2 \int_{\R^n} \h(x) \abs{v_m(x)}^2\,dx \\
& \leq   \int_{\R^n}  \left( \frac \b 2 - V_2(x) - \b_m\right) \h(x) \abs{v_m(x)}^2\,dx + \littleo m {+\infty} (1) \\
& \leq  \left( \frac \b 2 +m_- \right) \int_{\R^n} \h(x) \abs{v_m(x)}^2\,dx + \littleo m {+\infty} (1) .
\end{align*}
%
%
%
%
%
%
This proves that

\[
\int _{\R^{2n}} \h(x) \, d\m(x,\x) = \lim_{m \to \infty} \int_{\R^n} \h(x) \abs {v_m(x)}^2 \, dx  \neq 0.
\]
\end{proof}

We now prove that $\m$ is actually zero. Note that by Proposition \ref{prop-burq-debut} we already know that $\m$ is supported on $p\inv(\singl E)$.

\begin{proposition} \label{prop-masse-finie}
The total measure of $\m$ is finite.
\end{proposition}

\begin{proof}
Let $q \in C_0^\infty(\R^{2n},[0,1])$. We have
\[
 \nr{\Opw(q)}_{\Lc(L^2(\R^n))} \leq C \nr q _{L^\infty(\R^{2n})} + \bigo m \infty \big( \sqrt {h_m} \big) 
\]
where $C$ only depends on the dimension $n$ (see for instance Theorem 5.1 in \cite{evansz}), and hence:
\[
\int_{\R^{2n}} q  \, d\m = \lim _{m\to \infty} \innp{\Opwm (q ) v_m}{v_m} \leq  \limsup_{m \to \infty } C\nr q _{L^\infty(\R^{2n})} \nr{v_m}_{L^2(\R^n)}^2 \leq C.
\]
Considering $q$ equal to 1 on the ball $B_{\R^{2n}}(k)$ of radius $k \in\N$ in $\R^{2n}$ proves that ${\m(B_{\R^{2n}}(k))\leq C}$ for all $k\in\N$. 
\end{proof}

\begin{proposition}  \label{prop-vp-infty}
$\m = 0$ on $\O_\infty^-(\singl E)$.
\end{proposition}

\begin{proof}
Let $q \in C_0^\infty(\R^{2n},[0,1])$ supported in $\O_\infty^-(\R_+^*)$. There exists $T \geq 0$ such that for $w \in \supp q$ and $s \geq T$ we have $V_2(\bar x (-s,w)) + \b \geq \frac \b 2$. Put
\[
C =  \sup_{ \supp q}  \exp \left(- 2 \int_0^T (V_2 +\b)\circ \vf^{-s} \, ds \right).
\]
According to \eqref{propagation-mu} and Proposition \ref{prop-masse-finie}, we have for all $t \geq T$:
\begin{align*}
0 \leq \int_{\R^{2n}} q \, d\m
& = \int _{\R^{2n}} (q\circ \vf^t) \exp \left( -2\int_0^t (V_2 + \b) \circ \vf^{t-s} \, ds \right) \, d\m\\
& \leq \m(\R^{2n}) \sup_{ \supp q} \exp \left( -2\int_0^t (V_2 + \b) \circ  \vf^{-s} \, ds\right) \\
& \leq C \m(\R^{2n}) \exp \left(-(t-T)  \b  \right) \\
& \limt t \infty 0.
\end{align*}
This implies that the integral of $q$ is zero and proves the proposition.
\end{proof}

\begin{proposition} \label{prop-vp-ob}
$\m = 0$ on $\O_b^-(\singl{E})$ and hence on $\R^{2n}$.
\end{proposition}

\begin{proof}
We follow the idea of the previous proof, now using the absorption assumption on trapped trajectories. Let $q \in C_0^\infty(\R^{2n},[0,1])$. Since $\m$ is supported on  $\O_b^-(\singl{E})$ we have for all $t \geq 0$:
\begin{align*}
0 \leq \int_{\R^{2n}} q \, d\m
& = \int _{\R^{2n}} (q\circ \vf^t) \exp \left( -2\int_0^t (V_2 + \b) \circ \vf^{t-s} \, ds\right) \, d\m\\
& \leq \m(\R^{2n}) \sup_{\O_b^-(\singl E) \cap \supp q} \exp \left(-2 \int_0^t (V_2 + \b) \circ \vf^{-s} \, ds\right) .
\end{align*}
Now using Proposition \ref{prop-gros-amort}, we can conclude that the integral of $q$ is zero.
\end{proof}

Propositions \ref{prop-vp-nonzero} and \ref{prop-vp-ob} give the contradiction which proves Theorem \ref{th-vp}. We remark that assumption \eqref{hyp-amfaible} is stronger than necessary to prove Proposition \ref{prop-vp-ob}, since we did not use the fact that $\Im z_m$ is allways greater that $\b$. We can actually prove the following result:

\begin{corollary} \label{cor-vp}
Let $E > 0$ and $\b > 0$ such that
\[
\forall w \in \O_b(\singl E) , \exists T \geq 0, \quad \int_0^T \big(V_2(\bar x (-s,w)) +\b \big)   \, ds > 0.
\]
Then there exist a neighborhood $I$ of $E$, $h_0 > 0$ and $c \geq 0$ such that for $h \in ]0,h_0]$ the operator $\hh$ has no eigenvalue in $\C_{I,h\b}$ and
\[
\forall z \in \C _{I,h\b}, \quad \nr{(\hh-z)\inv}_{\Lc(L^2(\R^n))} \leq \frac c {h}.
\]
\end{corollary}

As explained before Corollary \ref{cor-estim}, we obtain results on the lower half-plane under the more ``symmetric'' non-trapping condition:

\begin{corollary}
Let $E > 0$ be a non-trapping energy and $\b > 0$. Then there exist a neighborhood $I$ of $E$, $h_0 > 0$ and $c \geq 0$ such that for $h \in ]0,h_0]$, $\Re z \in I$ and $\abs {\Im z } \geq h \b$ the operator $(\hh-z)$ has a bounded inverse and
\[
\nr{(\hh-z)\inv}_{\Lc(L^2(\R^n))} \leq \frac c {h}.
\]
\end{corollary}

\begin{proof}
We only have to apply Theorem \ref{th-vp} to $\hh$ and its adjoint $\hh^*$.
\end{proof}

\section{Estimate in the incoming region}

In this section we prove an estimate for the outgoing solution of the Helmholtz equation in the incoming region. 
Assume that $V_2$ is of short range. Let $I\subset \R$, $h_0 > 0$, $\d > \frac 12$, $c\geq 0$, $k\in\N$ and suppose that for $h \in \ ]0,h_0]$ and $z \in \C_{I,+}$ the resolvent $(\hh-z)\inv$ is well-defined and
\begin{equation} \label{hyp-estim}
\nr{\pppg x ^{-\d} (\hh -z)\inv \pppg x ^{-\d}}_{\Lc(L^2(\R^n))} \leq \frac c {h^k}.
\end{equation}

\begin{theorem} \label{th-incoming}
Let $R_1 > 0$, $d > d_1 \geq  0$ and $\s > \s_1 \geq 0$. Then there exists $R > R_1$ such that for $z \in  {\C_{I,+}}$, $\b \in \R$, $\o \in \symb_0(\R^{2n})$ supported outside $\zoneS_- (R_1,d_1,-\s_1)$ and $\o_- \in \symb_0 (\R^{2n})$ supported in $B_x(r) \cap \zoneS_-(R,d,-\s)$ (for some $r > 0$), we have
\[
\nr{\pppg x ^{\b} \Op(\o_-) (\hh-z)\inv  \Op (\o) \pppg x ^{\b}}_{\Lc(L^2(\R^n))} = \bigo h 0 (h^\infty),
\]
and the size of the rest is uniform in $z \in  {\C_{I,+}}$. Moreover if the limiting absorption principle holds in $\Lc(L^{2,\d}(\R^n),L^{2,-\d}(\R^n))$ for $\l\in I$ and $h>0$ small enough, then the estimate remains true for $(\hh-(\l +i0))\inv$, $\l \in I$.
\end{theorem}

\begin {remark}
We are going to use this result for a small perturbation of a dissipative Schrödinger operator in order to prove Theorem \ref{th-estim} (see Proposition \ref{prop-mzero-incoming}). Then, once Theorem \ref{th-estim} is proved, we can use Theorem \ref{th-incoming} for the full non-dissipative Schrödinger operator $\hh$ we are interested in (see Section \ref{sec-mesure}).
\end {remark}

The result being stronger in this case, we can assume without loss of generality that $d_1$ and $\s_1$ are positive. The proof of this theorem follows that of the dissipative analog given in \cite{art_mesure}. We recall the sketch of the proof for the reader convenience and refer to \cite{robertt89, wang88, art_mesure} for more details.\\

Let $d_0 \in ]0,d_1[$ and $\s_0 \in ]0,\s_1[$. There exist $R_ 0 > 0$ and $\vf \in C^\infty(\R^{2n})$ such that
\begin{equation} \label{izosaki-kitada}
\forall (x,\x) \in \zoneS_-(R_0,d_0,-\s_0), \quad \abs{ \nabla_x \vf (x,\x)} ^2 + V_1 (x) = \abs \x ^2
\end{equation}
and, for some $\rho > 0$:
\[
 \forall (x,\x) \in \R^{2n},\forall \a,\b \in \N^n, \quad  \abs{\partial_x^\a \partial_\x^\b \big( \vf(x,\x) - \innp x \x  \big)} \leq c_{\a,\b}\pppg x ^{1 - \rho - \abs \a}
\]
(see \cite{isozakik85}). As explained in \cite{wang88}, we can assume that the constants $c_{\a,\b} > 0$ for $\a , \b\in \N^n$ are as small as we wish as long as we replace $\vf$ by
\begin{equation} \label{newphi}
(x,\x) \mapsto \big(\vf(x,\x) - \innp x \x\big) \h \left( \frac x R \right ) + \innp x \x
\end{equation}
for $R \geq 2 R_0 $ large enough and $\h \in C^\infty(\R^n)$ such that $\h(x)=0$ for $\abs x \leq 1/4$ and $\h(x) = 1$ for $\abs x \geq 1/2$. In this case \eqref{izosaki-kitada} remains valid on $\zoneS_-(R/2,d_0,-\s_0)$.\\

For all $(x,\x) \in \R^{2n}$ we denote by $t \mapsto r(t,x,\x) \in \R^n$ the solution of the problem
\[
 \left\{ \begin{array} l \partial_t r (t,x,\x) = \nabla _x \vf \big(r (t,x,\x),\x\big), \\ r(0,x,\x) = x . \end{array} \right.
\]
We can check that this defines a smooth function on $\R \times  \R^{2n}$. If $R$ was chosen large enough in \eqref{newphi}, then for $(x,\x) \in \zoneS_-(0,d_1,-\s_1)$ and $t\geq 0$ we have
\[
 \abs {r(-t,x,\x)} \geq \abs x + \frac {\s_1 d_1 t} 2 .
\]
Moreover for $\a,\b \in \N^n$ with $\abs \a + \abs \b \geq 1$ there exists $c_{\a,\b} \geq 0$ such that
\[
 \abs { \partial_x^\a \partial_\x^\b r(-t,x,\x)} \leq c_{\a,\b} (t + \pppg x ) \pppg x ^{-\abs \a}.
\]

For $t\geq 0$ and $(x,\x) \in \zoneS_-(0,d_1,-\s_1)$ we now set
\[
F (t,x,\x) = \D_x\vf(r(t,x,\x),\x) - V_2(r(t,x,\x)
\]
in order to define on $\zoneS_-(0,d_1,-\s_1)$ the symbols
\[
a_{0}(x,\x) = \exp\left({- \int _{0}^\infty F(- 2s,x,\x)}\, ds \right)
\]
and, for $j \geq 1$:
\[
 a_{j}(x,\x) =  i \int_{0 }^{+\infty} \D_x a_{j-1} \big(r(- 2\t,x,\x),\x \big) \exp\left( - \int_0^\t F(- 2s,x,\x)\,ds \right) \, d\t.
\]
These functions are solutions of the transport equations
\[
2\nabla_x a_{0} \cdot \nabla_x \vf + a_{0} \D_x \vf  - a_{0} V_2 = 0
\]
and, for $j \geq 1$:
\[
  2 \nabla_x a_{j} \cdot \nabla_x \vf + a_{j} \D_x \vf - a_{j} V_2 - i \D_x a_{j-1} = 0.
\]
Moreover $a_j$ decays as a function of $\symb_{-j}(\R^{2n})$ and there exists $c_0 > 0$ such that
\[
\forall (x,\x) \in \zoneS_-(0,d_1,-\s_1), \quad \abs{a_0(x,\x) } \geq c_0.
\]
Note that $V_2$ has to be of short range here but the sign does not matter. \\

Since we work on $\zoneS_-(0,d_1,-\s_1)$, we now introduce a cut-off function as follows. We choose $R_2$ and $R_3$ such that $\max (R_1 , R/2) < R_2 < R_3<R$, $d_2$ and $d_3$ such that $d_1 < d_2 < d_3 < d$ and finally $\s_2$ and $\s_3$ such that $\s_1 < \s_2 < \s_3 < \s$. Then we consider $\h_1,\h_2 , \h_3 \in C^\infty(\R,[0,1])$ such that $\h_1(s) = 0 $ for $s \leq R_2$, $\h_1(s) = 1$ for $s \geq R_3$, $\h_2(s) = 0$ for $s \leq d_2$, $\h_2(s) =1$ for $s \geq d_3$, $\h_3(s) = 0$ for $s \leq \s_2$ and $\h_3(s) = 1$ for $s \geq \s_3$. We fix $N\in\N$. Let us define
\[
a(h) = \sum_{j=0}^N h^j a_j\quad  \text{and} \quad b(x,\x,h) =   \h_1 (\abs x) \h_2 (\abs \x) \h_3 \left(- \frac { x \cdot \x}{\abs x \abs \x}  \right)  a(x,\x,h).
\]
We also consider
\[
p(h) = \frac ih \big( \abs{\nabla _x \vf }^2 + V_1 - \abs \x^2 \big) b(h) + \big( 2 \nabla_x b(h) \cdot \nabla_x \vf  + b(h) \D_x \vf - b(h) V_2 \big) -ih \D_x b(h).
\]
The symbols $b(h)$ and $p(h)$ are supported in $\zoneS_-( R_2 ,d_2 , -\s_2)$ and for $\a,\b\in\N^n$ there exists a constant $c_{\a,\b} >0$ such that for $\hou$ we have
\[
\forall (x,\x) \in \zoneS_-( R_2 ,d_2 , -\s_2), \quad \abs{\partial_x^\a \partial_\x^\b b(x,\x,h)} +\abs{\partial_x^\a \partial_\x^\b p(x,\x,h)} \leq c_{\a,\b} \pppg x ^{- \abs \a}
\]
and
\[
\forall (x,\x) \in \zoneS_-( R_3 ,d_3 , -\s_3), \quad \abs{\partial_x^\a \partial_\x^\b p(x,\x,h)}  \leq c_{\a,\b} \, h^{N+1} \pppg x ^{-2-N- \abs \a} .
\]

If $R$ is chosen large enough, $R_5 \in ]R_3, R[$, $d_5 \in ]d_3,d[$ and $\s_5 \in ]\s_3,\s[$, then we can construct (see \cite[Lemma 4.5]{wang88}) a symbol $e(h) =\sum_{j=0}^N h^j e_j$ such that $e_j \in \symb_{-j}(\R^{2n})$ is supported in $\zoneS_-(R_5,d_5,-\s_5)$ for all $j \in \Ii 0 N$ and
\[
 I_h (e(h), \vf) I_h (b(h),\vf)^* = \Op(\o_-)  + h^{N+1} \Op(r(h)),
\]
where $r(h) \in \symb_{-N}(\R^{2n})$ uniformly in $h \in ]0,1]$ and for $u \in \Sc(\R^n)$ we have set
\[
I_h(b,\vf) u(x)  = \frac 1 {(2\pi h)^n} \int_{\R^n}\int_{\R^n} e^{\frac ih ( \vf(x,\x) - \innp y \x)} b(x,\x) u(y)\,d\x \, dy.
\]


For any $t \geq 0$ we have
\[
I_h(b(h),\vf) ^* \uh (t) = \uoh (t) I_h(b(h),\vf)^* - \int_0^t \uoh(s) I_h(p,\vf)^* \uh(t-s) \, ds ,
\]
and hence
\begin{align*}
\Op(\o_-) \uh(t)
& = -h^{N+1} \Op(r(h)) \uh(t) + I_h (e(h),\vf) \uoh(t) I_h(b(h),\vf)^*\\
& \quad - \int_0 ^t I_h(e(h),\vf) \uoh(s) I_h(p(h),\vf) \uh(t-s) \, ds.
\end{align*}
Contrary to the dissipative case, we cannot write
\[
 (\hh-z) \inv = \frac ih \int_0^\infty e^{-\frac{it} h (\hh-z)} dt =\frac ih \int_0^\infty e^{\frac {it}h z} U_h(t) \,  dt
\]
for any $z \in \C_+$, but only when $\Im z > h\nr {V_2}_\infty$. For such a $z$, we obtain from the previous equality:
\begin{eqnarray*}
\lefteqn{\pppg x ^{\b} \Op(\o_-)(\hh-z)\inv \Op(\o) \pppg x ^{\b}} \\
&& = - h^{N+1}\pppg x ^{\b} \Op(r(h)) (\hh-z)\inv \Op(\o) \pppg x ^{\b} \\
&& \quad  + \frac ih \pppg x ^{\b}\int_{0}^\infty e^{\frac{it}h z}  I_h(e(h),\vf) U_0^h(t)I_h(b(h),\vf)^* \Op(\o) \pppg x ^{\b}\, dt \\
&& \quad -  \pppg x ^{\b}\int_{0}^\infty e^{\frac{is}h z}   I_h(e(h),\vf) U_0^h(s)I_h(p(h),\vf)^ *  (\hh-z)\inv \Op(\o)  \pppg x ^{\b}\, ds.
\end{eqnarray*}
This equality is proved for $\Im z > \nr {V_2}_\infty$ but the two integrands decay with time uniformly in $z \in \C_{I,+}$. Each term is holomorphic on $\C_{I,+}$, so for any $h \in ]0,h_0]$ this equality remains valid on $\C_{I,+}$ by unique continuation. Then it only remains to estimate each term of the right-hand side to conclude (we use assumption \eqref{hyp-estim} here). This can be done as in the dissipative case.

\section{Uniform resolvent estimates}

We now prove the uniform resolvent estimates for the non-dissipative Schrödinger operator up to the real axis.
In order to use Theorem \ref{th-incoming} we assume that $V_2$ is of short range. But we expect Theorem \ref{th-estim} to be true under a weaker assumption on $V_2$, so we are going to give the other arguments only assuming that $V_2 \in C^\infty(\R^n)$ is of long range with a short range negative part: there exist $\rho > 0$, $C > 0$ and constants $c_\a$ for $\a \in \R^n$ such that for all $x\in \R^n$ we have
\begin{equation} \label{estim-V2}
V_2(x) \geq - C \pppg x ^{-1-\rho} \quad \text{and} \quad  \forall \a \in \N^n,\, \abs {\partial^\a V_2(x) } \leq c_\a \pppg x ^{-\abs \a -\rho}.
\end{equation}

As for Theorem \ref{th-vp}, we proceed by contradiction. We suppose that Theorem \ref{th-estim} is wrong and consider sequences $\seq v m \in H^2(\R^n)^\N$, $\seq z m \in \C^\N$ and $\seq h m \in ]0,1]^\N$ such that if we set $\l_m = \Re z_m$ and $\b_m = h_m\inv \Im z_m$ then for some $\d \in \left] \frac 12 , \frac {1+\rho}2 \right[$ we have
\[
h_m \to 0, \quad  \l_m \to E, \quad  0 < \b_m  \to 0 ,  \quad \nr {v_m}_{L^{2,-\d}(\R^n)} = 1
\]
and
\[
\nr{(\hhm -z_m) v_m}_{L^{2,\d}(\R^n)} = \littleo m \infty (h_m).
\]
We remark that $v_m$ is assumed to be in $H^2(\R^n)$ for all $m\in \N$, but is only uniformly bounded in $L^{2,-\d}(\R^n)$.\\

We are going to prove that such a sequence $\seq v m$ cannot exist. First considering a sequence of eigenvectors, this will prove that for $h$ small enough, $\hh$ has no eigenvalue with real part close to $E$ and positive imaginary part. Then, the operator $\pppg x ^{-\d} (\hh-z)\inv \pppg x ^{-\d}$ is well-defined as a bounded operator from $L^2(\R^n)$ to $H^2(\R^n)$ when $\Im z > 0$, $\Re z$ is close to $E$ and $h$ is small enough. Applying again the argument now gives the estimate of Theorem \ref{th-estim}.\\

If there exists a subsequence $\seq m k$ such that $\b_{m_k} \geq \b > 0$ for all $k \in \N$, then we obtain a contradiction with Theorem \ref{th-vp}. Therefore we can assume that
\[
 \b_m \limt m \infty 0.
\]
After extracting a subsequence if necessary, we can assume that \eqref{conv-mesure} holds for some non-negative Radon measure $\m$. We already know that $\m$ is supported in $p\inv(\singl E)$. In order to get a contradiction, we prove that $\m = 0$ and $\m \neq 0$.\\

Let $W_2 = V_2  + 2 C \pppg x ^{-1-\rho} \geq C \pppg x ^{-1-\rho}$, the constant $C$ being given by \eqref{estim-V2}. We first prove that $\m \neq 0$. The proof relies on the existence of an escape function in the sense of \cite{jecko04}.

\begin{proposition} \label{prop-fonc-fuite}
Let $E > 0$. There exist $f \in \symbor(\R^{2n},\R)$, $\h \in C_0^\infty(\R^n,[0,1])$ and $\tilde \h \in C_0^\infty(\R,[0,1])$ equal to 1 in a neighborhood of $E$ such that
\[
\forall (x,\x) \in \R^{2n}, \quad \{p,f\} (x,\x) = (1-\h(x)) \tilde \h(p(x,\x)) \pppg x^{-2\d}.
\]
\end{proposition}

The proof of this proposition is postponed to appendix A.

\begin{proposition} \label{prop-mu-neq0}
The measure $\m$ is non-zero.
\end{proposition}

\begin{proof}
{\bf 1.}
Let $\th \in C_0^\infty(\R^n,[0,1])$ be supported in $B_2$ and equal to 1 on $B_1$. For $R > 0$ we set $\th_R(x) = \th\big( \frac x R \big)$. If there exists $R > 0$ such that
\begin{equation} \label{thetaR}
 \int_{\R^{2n}} \th_R(x) \pppg x ^{-1-\rho} \, d\m(x,\x) \neq 0,
\end{equation}
then the proposition is proved. Otherwise for any $R > 0$ we have
\begin{eqnarray*}
\lefteqn{\limsup_{m \to \infty} \abs{\innp{\pppg x ^{-1-\rho} v_m} {v_m}_{L^2(\R^n)}}}\\
&& \leq \limsup_{m \to \infty} \abs{  \innp{\th_R \pppg x ^{-1-\rho}v_m} {v_m}_{L^2(\R^n)} } + \limsup_{m \to \infty} \abs{\innp{(1- \th_R) \pppg x ^{-1-\rho} v_m} {v_m}_{L^2(\R^n)}}\\
&& \leq \nr{\pppg x ^{2\d-1-\rho} (1-\th_R) }_{L^\infty(\R^n)} \leq \pppg R^{2\d -1-\rho}.
\end{eqnarray*}
This proves that
\[
\innp{\pppg x ^{-1-\rho}v_m} {v_m}_{L^2(\R^n)} \limt m {+\infty} 0,
\]
and hence:
\begin{eqnarray} \label{bmV2}
\lefteqn{\b_m \nr{v_m}^2_{L^2(\R^n)} + \nr{\sqrt {W_2} v_m}^2_{L^2(\R^n)}}\\
\nonumber && = - h_m \inv \Im \innp{ (\hh-z_m) v_m}{v_m}_{L^2(\R^n)} + 2C \innp{\pppg x ^{-1-\rho} v_m}{v_m}_{L ^2(\R^n)}  \limt m {+\infty} 0.
\end{eqnarray}
Since both terms of the left-hand side are non-negative, this means that each goes to 0 as $m$ goes to $+\infty$. Moreover $\sqrt {W_2}$ is smooth and all its derivatives of order at least 1 belong to $\symb \big ( \pppg x ^{-\d} \big)$, so for any $f \in \symbor (\R^{2n})$ we have
\begin{equation} \label{W2fv}
 \nr {\sqrt {W_2} \Opw(f) v_m}_{L^2(\R^n)} = \nr {\Opw(f) \sqrt {W_2} v_m}_{L^2(\R^n)} + \bigo m \infty(h_m) \limt m \infty 0.
\end{equation}

\noindent
{\bf 2.}
Let $\h \in C_0^\infty(\R^n,[0,1])$, $\tilde \h \in C_0^\infty(\R,[0,1])$ and $f \in \symbor(\R^{2n},\R)$ as given by Proposition \ref{prop-fonc-fuite}. For $m\in\N$ we have:
\begin{align*}
1
& = \innp{ \pppg x^{-2\d} v_m}{v_m}\\
& = \innp{ \pppg x^{-2\d} \h(x) v_m}{v_m} + \innp{\Opwm  \big( \pppg x ^{-2\d} (1-\h(x)) ((1-\tilde \h)\circ p)\big) v_m }{v_m}\\
& \quad +  \innp{\Opwm  \big( \pppg x ^{-2\d} (1-\h(x)) (\tilde \h\circ p)\big) v_m }{v_m}.
\end{align*}
According to Proposition \ref{prop-burq-debut}, the second term goes to 0 as $m$ goes to $+\infty$. We now prove that this also holds for the third term to prove that
\[
 \int_{\R^{2n}} \pppg x ^{-2\d} \h(x) \, d\m(x,\x) = \lim_{m \to \infty} \innp{ \pppg x^{-2\d} \h(x) v_m}{v_m} \neq 0.
\]

\noindent
{\bf 3.}
We have
\[
\frac i {h_m} [\huhm , \Opwm(f)] =  \Opwm (\{p,f\})  + h_m^2 \Opwm(r_3(h_m)),
\]
where $r_3(h) \in \symb\big(\pppg x ^{-1-\rho}\big)$ uniformly in $h \in ]0,1]$, and hence
\begin{align*}
\innp{\Opw\big((\tilde \h \circ p) (1-\h(x)) \pppg x^{-2\d}\big) v_m}{v_m}
& = \innp{\Opwm(\{p,f\}) v_m}{v_m}\\
& = \frac i {h_m} \innp{[\huh , \Opwm(f)] v_m}{v_m} + \bigo m \infty(h_m^2).
\end{align*}
Since $v_m \in H^2(\R^n)$ for all $m \in\N$ and according to \eqref{bmV2}-\eqref{W2fv} we have
\begin{eqnarray} \label{conv-infinity}
\lefteqn{\innp{\Opw\big((\tilde \h \circ p) (1-\h(x)) \pppg x^{-2\d}\big) v_m}{v_m}}\\
\nonumber && = \frac i {h_m} \innp{\left( (\hhm-z_m)^* \Opwm(f) - \Opwm(f) (\hhm-z_m) \right) v_m}{v_m} + \littleo m \infty(1)\\
\nonumber&& \limt m \infty 0,
\end{eqnarray}
which concludes the proof.
%
%
\end{proof}

Let
\[
 \hdh = -h^2 \D + V_1(x) - ih W_2 (x).
\]
The Schrödinger operator $\hdh$ is dissipative and its dissipative part is positive on trapped trajectories of energy $E$, so there exist a neighborhood $I$ of $E$, $h_0 >0$ and $C_2 >0$ such that for $h \in ]0,h_0]$ and $z \in\C_{I,+}$ we have
\[
 \nr{\pppg x^{-\d} (\hdh -z)\inv \pppg x ^{-\d}}_{\Lc(L^2(\R^n))} \leq \frac {C_2} h
\]
(see \cite{art_mourre}). Since $2\d <1+\rho$ we can write $2C \pppg x^{-1-\rho} = W_3 + W_4$ where $W_4 \in C_0^\infty(\R^n)$ and
\[
 \forall x \in \R^n, \quad \pppg x^{2\d} \abs{W_3(x)} \leq \frac {1} {2C_2}.
\]
Put
\[
\hth = -h^2 \D + V_1(x) - ih W_2(x) +  ih W_3(x) = \hdh + ih W_3(x).
\]
Let $z \in \C_{I,+}$. By a standard perturbation argument, we know that for $h \in ]0,h_0]$ the resolvent $(\hth -z)\inv$ is well-defined and
\begin{equation}\label{estim-hth}
 \nr{\pppg x ^{-\d} (\hth -z)\inv \pppg x ^{-\d}}_{\Lc(L^2(\R^n))} \leq \frac {2C_2} h.
\end{equation}
As a result we can apply Theorem \ref{th-incoming} with $\hth$ on $\C_{I,+}$ (we recall that $V_2$ has to be of short range here). We use this result to prove that the semiclassical measure $\m$ is supported outside $\O_\infty^-(\singl E)$:

\begin{proposition} \label{prop-mzero-incoming}
$\m = 0$ on $\O_\infty^-(\singl{E})$.
\end{proposition}

\begin{proof}
Let $J \subset I \cap ]E/2,2E[$ be a neighborhood of $E$. We first check that $\m=0$ in the incoming region $\zoneS_-(R,0,-1/2)$ for some $R$ large enough. Let $R_1$ be such that $\supp W_4 \subset B_{R_1}$, $d \in ]0, \sqrt{E/2}[$ and $\s = \frac 12$. Let $R$ be given by Theorem \ref{th-incoming} applied to $\hth$, and finally $\o_- \in C_0^\infty(\R^{2n})$ supported in $\zoneS_-(R,d,-1/2)$. For $m$ large enough the operator $(\hthm -z_m)$ has a bounded inverse, so we can write
\begin{align*}
\Opwm(\o_-) v_m
& =  \Opwm(\o_-) (\hthm -z_m)\inv (\hthm -z_m) v_m\\
& = \Opwm(\o_-) (\hthm -z_m)\inv (\hhm -z_m)v_m\\
& \quad  - ih_m \Opwm(\o_-)(\hthm -z_m)\inv W_4 v_m.
\end{align*}
According to \eqref{estim-hth} and Theorem \ref{th-incoming} we obtain
\begin{eqnarray}
\label{calc-zone-entr}
\lefteqn{\nr{\Opwm(\o_-) v_m}_{L^{2,\d}(\R^n)}}\\
\nonumber
&& =  \nr{ \pppg x ^\d \Opw(\o_-) \pppg x ^\d } \nr{\pppg x ^{-\d} (\hthm -z_m)\inv \pppg x ^{-\d}} \nr{(\hhm-z_m)v_m}_{L^{2,\d}(\R^n)}\\
\nonumber
&&\quad  +  h_m \nr{\pppg x ^\d\Opw(\o_-) (\hthm -z_m)\inv  W_4 \pppg x ^\d}   \nr{v_m}_{L^{2,-\d}(\R^n)}\\
\nonumber
&& \limt m \infty 0.
\end{eqnarray}
This proves that
\[
\innp{\Opwm(\o_-) v_m}{v_m}_{L^2(\R^n)} \limt m \infty 0,
\]
and hence $\m = 0$ on $\zoneS_-(R,d,-1/2)$. Now let $q \in C_0^\infty(\R^{2n})$ supported in $\O_\infty^-(J)$. For $t \geq 0$ large enough we have $\vf^{-t} (\supp q) \subset \zoneS_-(R,d,-1/2)$. According to Proposition \ref{prop-burq-debut} we obtain
\[
 \int_{\R^{2n}} q \, d\m = \int_{\R^{2n}} (q \circ \vf^t) \exp \left( -2 \int_0^t V_2 \circ \vf^s \, ds \right) \, d\m = 0,
\]
which proves that $\m = 0$ on $\O_\infty^-(J)$.
\end{proof}

When \eqref{thetaR} holds for some $R >0$, then Proposition \ref{prop-mu-neq0} is proved but not \eqref{conv-infinity}. So we cannot use it to show that $\m$ is zero at infinity as is done in the self-adjoint case (see \cite{jecko04}).

\begin{proposition}
$\m = 0$ on $\O_b^-(\singl{E})$.
\end{proposition}

\begin{proof}
This proposition is proved as Proposition \ref{prop-vp-ob}, even though the total measure of $\m$ is no longer necessarily finite. Let $q \in C_0^\infty(\R^{2n},[0,1])$. We know that $\m$ is supported in $\O_b^-(\singl{E})$, so according to Proposition \ref{prop-burq-debut} we can write
\[
\int_{\R^{2n}} q \, d\m = \int_{\R^{2n}} \1 {\O_b^-(\singl E)} (q\circ \vf ^t) \exp\left( -2 \int_0^t V_2 \circ \vf^{t-s} \, ds\right) d\m.
\]
Since the set
\[
\bigcup _{t\geq 0} \vf^{-t}(\supp q \cap \O_b^-(\singl E))
\]
is bounded, there exists $c \geq 0$ such that for all $t \geq 0$
\[
\int_{\R^{2n}} q \, d\m \leq c \sup_{\supp q \cap \O_b^-(\singl E)}   \exp\left( -2 \int_0^t V_2 \circ \vf^{-s} \, ds\right).
\]
Then we can conclude with Proposition \ref{prop-gros-amort}.
\end{proof}

\section{Limiting Absorption Principle}

After having proved resolvent estimates on $\C_{I,+}$, we can show the limiting absorption principle and study existence and uniqueness of an outgoing solution for \eqref{helmholtz}. Before giving more precise statements, we introduce some notation. 
Let
\[
 \C_{++} = \singl{ \z \in \C \tqe \Re \z > 0, \Im \z \geq 0}.
\]
 For $u \in H^1(\R^n)$ supported outside a neighborhood of 0 we set
\[
\partial_r u = \frac {x \cdot \nabla u} {\abs x}, \quad    \dpar u = \partial_r u  + \frac {n-1}{2\abs x}u \quad \text{and} \quad \nabla_\bot u  = \nabla u  -  \frac {x \partial_r u} {\abs x}.
\]
We are going to use the following basic properties of these operators:
\begin{proposition}  \label{prop-dpar}
For $u,v\in H^1(\R^n)$ supported outside a neighborhood of 0 we have
\[
\frac d {dr}  \innp u v _\ldgr  = \innp{\dpar u} v _\ldgr  + \innp u {\dpar v}_\ldgr,\\
\]
and
\[
\innp{\partial _r u}{\partial_r v} = \innp{\dpar u}{\dpar v} + \innp{ \frac {(n-1)(n-3)}{4\abs x^2} u} v.
\]
If moreover $u$ belongs to $H^2(\R^n)$ we also have
\[
\partial_r \nabla_\bot u(x) = - \frac 1 {\abs x} \nabla _\bot u(x) + \nabla_\bot \partial_r u(x) .
\]
\end{proposition}

In this section we consider a Schrödinger operator
\[
 H = -\D + V_1(x) -i V_2(x)
\]
with domain $\Dom(H) = H^2(\R^n)$. $V_1$ and $V_2$ are bounded and real-valued. We assume that 
\begin{equation} \label{dec-pot}
 V_1 -iV_2 = W_1 -iW_2 + W_3,
\end{equation}
where: 
\begin{enumerate}[(i)]
\item $W_1$ and $W_2$ are differentiable,
\item for all $x \in \R^n$ we have $W_1(x) \in \R$, $W_2(x) \geq 0$ and $W_3(x) \in \C$,
\item there exist $\rho >0$ and $c\geq 0$ such that for all $x \in \R^n$ we have
\begin{equation} \label{dec-W}
 \abs{W_1(x)} + W_2(x) \leq c\pppg x ^{-\rho} \quad \text{and} \quad \abs{\nabla W_1(x)} + \abs{\nabla W_2(x)} + \abs{W_3(x)} \leq c \pppg x ^{-1-\rho}.
\end{equation}
\end{enumerate}
For $x \in B_1^c$ we also set
\[
 \tilde W_3(x) = W_3 (x) + \frac {(n-1)(n-3)}{4 \abs x^2}.
\]

\begin{definition} \label{def-outgoing}
 Let $\z \in \C_{++}$, $f \in L^2_\loc(\R^n)$, and suppose that $u \in H^2_\loc(\R^n)$ is a solution for the equation
\begin{equation} \label{eq-radiation}
 (H-\z^2) u = f.
\end{equation}
Then we say that $u$ is an outgoing solution for \eqref{eq-radiation} if there exists $\d > \frac 12$ such that $(\dpar-i\z) u \in L^{2,\d-1}(B_1^c)$.
\end{definition}

Let $\d \in \left] \frac 12, \frac 12 + \frac \rho 4 \right[$ be fixed for all this section, $\rho$ begin given by \eqref{dec-W}. 
Let $K$ be a compact subset of $\C_{++}$ such that $K = \bar {K \cap \C_+}$. We set $K^* = K \cap \C_+$.

\begin{proposition} \label{prop-lap}
Assume the resolvent $(H-\z^2)\inv$ is defined for all $\z \in  K^*$ and the equation $(H-\l^2)u=0$ has no non-trivial outgoing solution when $\l \in K\cap \R_+^*$. Let $\l \in K \cap \R_+^*$ and $f \in  L^{2,\d}(\R^n)$. 

Then $(H-\z^2)\inv f $ converges in $L^{2,-\d}(\R^n)$ to the unique outgoing solution for the equation $(H-\l^2)u = f$ when $\z \in K^*$ goes to $\l$.

Moreover, there exists a constant $C \geq 0$ such that for any $\z \in K$ and $f \in L^{2,\d}(\R^n)$, if we denote by $u$ the unique outgoing solution for the equation $(H-\z^2)u = f$ then we have for all $R >0$ the following estimates:
\begin{equation} \label{estim-lap}
 \nr u _{L^{2,-\d}(\R^n)} + \nr{(\dpar -i\z)u}_{L^{2,\d-1}(B_1^c)} + R^{\d - \frac 12} \nr {u}_{L^{2,-\d}(B_R^c)}  \leq C \nr f _{L^{2,\d}(\R^n)}.
\end{equation}
\end{proposition}

We denote by $(H-(\l^2 + i 0))\inv f$ the unique outgoing solution $u \in H^2_{\loc}(\R^n) \cap L^{2,-\d}(\R^n)$ for the equation $(H-\l^2)u =f$.
To prove this theorem, we study the behavior of the solutions at infinity. In a compact subset of $\R^n$, the estimates we need are given by the interior regularity (see for instance \cite[§6.3.1]{evans}):

\begin{proposition} \label{prop-reg-int}
Let $f \in L_{\loc}^2(\R^n)$ and $z\in\C$. If $u \in H_{\loc}^2(\R^n)$ is a solution for the equation $(H-z)u=f$ then for all $R \geq 0$ we have
\[
\nr u _{H^2(B_R)} \leq C \left( \nr f _{B_{R+1}} + \nr u _{B_{R+1}} \right),
\]
where $C$ is uniform for $(z,R)$ in a compact subset of $\C \times \R_+$.
\end{proposition}

The difficulty is to give some estimates of $(H-\z^2)\inv f$ uniform when $\z$ approaches $\R_+^*$. For some fixed $\z \in K^*$ we have the following lemma:
\begin{lemma} \label{lem-est-udu}
 Let $f \in L^{2,\d}(\R^n)$ and $\z \in K^*$. Then $u$ and $\nabla u$ belong to $L^{2,\d}(\R^n)$.
\end{lemma}


The self-adjoint version of the following result is Lemma 4.1 in \cite{saito}.

\begin{lemma} \label{lem4.1}
There exists $C$ such that for $\z = \z_1 + i \z_2\in K^*$, $f \in L^{2,\d}(\R^n)$ and $u = (H-\z^2)\inv f$ we have
\begin{equation*}
\z_2 \nr{  u}  _{L^{2,1-\d}(\R^n) } \leq  C \left( \nr{u}_{L^{2,-\d}(\R^n)} +  \nr{(\dpar-i\z) u}_{L^{2,\d-1}(B_2^c)}  + \nr{f}_{L^{2,\d}(\R^n)}\right)
\end{equation*}
and
\[
\nr{\sqrt {W_2} u}^2_{L^{2,\frac 12 -\d}(\R^n)} \leq C \left( \nr{u}_{L^{2,-\d}(\R^n)} +  \nr{(\dpar - i\z) u}_{L^{2,\d-1}(B_2^c)}  + \nr{f}_{L^{2,\d}(\R^n)}\right) \nr{u}_{L^{2,-\d}(\R^n)}.
\]
\end{lemma}

\begin{proof} 
We consider  $\h : x \mapsto \tilde \h(\abs x)$ on $\R^n$, where $\tilde \h \in C^\infty(\R, [0,1])$ is equal to 0 on $]-\infty,2]$ and equal to 1 on $[3,+\infty[$. Let $\a \in [0,1-\d]$, $\z = \z_1 + i \z_2\in K^*$, $f \in L^{2,\d}(\R^n)$ and $u = (H-\z^2)\inv f$. According to Lemma \ref{lem-est-udu} $u$ belongs to $L^{2,\d}(\R^n)$, so we can write
\[
 \innp{(H-\z^2) u }{\h^2(x) (1+\abs x)^{2\a} u}_{L^2(\R^n)}  =  \innp{f}{\h^2(x) (1+\abs x)^{2\a} u}_{L^2(\R^n)}.
\]
Taking the imaginary part in this equality gives
\begin{eqnarray*}
\lefteqn{ \nr{\h \sqrt {W_2} u}^2_{L^{2,\a}(\R^n)} + 2 \z_1 \z_2 \nr {\h u}^2_{L^{2,\a}(\R^n)}}\\
&& \leq \nr f _{L^{2,\d}(\R^n)} \nr u _{L^{2,2\a - \d}(\R^n)}  + \nr u _{L^{2,-\d}(\R^n)} \nr {W_3 u}_{L^{2,2\a + \d}(\R^n)}\\
&& \quad + \nr{\h \partial_r u}_{L^{2,-\d}(\R^n)} \nr{ \big(2(1+\abs{x})^{2\a} \partial_r \h + 2\a (1+ \abs x)^{2\a-1}\h  \big) u}_{L^{2,\d}(\R^n)}\\
&& \leq c \left( \nr f _{L^{2,\d}(\R^n)}  + \nr{\h (\dpar-i\z) u}_{L^{2,\d-1}(\R^n)} +  \nr u _{L^{2,-\d}(\R^n)}\right) \nr { u}_{L^{2,2\a + \d-1}(\R^n)} . 
\end{eqnarray*}
Proposition \ref{prop-reg-int} and this inequality with $\a = 1-\d$ give the first estimate. We take $\a = \frac 12 - \d$ to obtain the second.
\end{proof}

For all $\z \in K^*$ and $x \in\R^n$ the complex number $\z^2 + iW_2(x)$ belongs to $\C_+$ and hence has a unique square root $\z_W (x,\z) = \z_1 (x,\z) + i \z_2(x,\z) \in \C_{++}$. Moreover there exists $C > 0$ such that for all $\z = \z_1 + i\z_2 \in K^*$ and $x \in\R^n$ we have
\begin{equation} \label{zeta2}
 C\inv \big(2 \z_1 \z_2 + W_2(x)\big) \leq \z_2(x,z) \leq C \big(2 \z_1 \z_2 + W_2(x)\big)
\end{equation}
and
\[
\abs{\nabla_x \z_W (x,\z)} \leq C \pppg x^{-1-\rho},
\]
where $\rho > 0$ is given by \eqref{dec-W}. For $\z \in K^*$ we set
\[
\dparz = \dpar - i \z_W(x,\z).
\]

\begin{proposition} \label{prop-dparz}
 There exists $C$ such that for $\z = \z_1 + i \z_2\in K^*$, $f \in L^{2,\d}(\R^n)$ and $u = (H-\z^2)\inv f$ we have
\[
 \nr{(\dpar -i\z) u}_{\nlpdsp{\d-1}{B_2^c}} + \nr{\nabla_\bot u }_{\nlpdsp {\d-1}{B_2^c}} \leq C \left( \nr u _{\nlpds {-\d}} + \nr f _{\nlpds \d} \right).
\]
\end{proposition}

The proof is inspired from the self-adjoint version given in \cite{saito}.

\begin{proof}
{\bf 1.} Since
\[
\abs{\z_W(x,\z) -\z} = \frac {W_2(x)}{\abs{\z_W(x,z) + \z}} \leq c W_2(x),
\]
we have
\begin{equation*} 
 \nr{(\dpar-i\z)u}_{L^{2,\d-1}(\R^n)} \leq  \nr{\dparz u}_{L^{2,\d-1}(\R^n)} + c \nr {u}_{L^{2,-\d}(\R^n)} ,
\end{equation*}
where $c \geq 0$ denotes different constants which do not depend on $\z \in K^*$ and $f \in L^{2,\d}(\R^n)$. It is therefore enough to prove the proposition with $(\dpar -i\z)$ replaced by $\dparz$.\\

\noindent
{\bf 2.} We consider $\h : x \mapsto \tilde \h (\abs x)$, where $\tilde \h \in C^\infty(\R,[0,1])$ is non-decreasing, equal to 0 on $]-\infty,1]$ and equal to 1 on $[2,+\infty[$. For $R > 2$ we set $\h_R (x) = \tilde \h(\abs x) - \tilde \h (\abs x - R)$.  Let $\z = \z_1 + i \z_2\in K^*$, $f \in L^{2,\d}(\R^n)$ and $u = (H-\z^2)\inv f$. It is enough to consider the case $\Nc_u < \Nc_d$, where 
\[
\Nc_u = \nr {u} _{L^{2,-\d}(\R^n)} +  \nr {f} _{L^{2,\d}(\R^n)} \quad \text{and} \quad \Nc_d = \max\left( \nr{\h \dparz u}_{L^{2,\d-1}(\R^n)} , \nr{\h \nabla_\bot u}_{L^{2,\d-1}(\R^n)} \right).
\]
According to Proposition \ref{prop-dpar} we can write
\begin{align*}
\Re \innp{f} {\h_R^2 (1+ \abs x)^{2\d}  \dparz u}
& = \Re\innp{(H-\z^2) u} {\h_R^2 (1+ \abs x)^{2\d-1}  \dparz u}\\
& = \Re\innp{\dpar u } { \dpar \h_R^2 (1+ \abs x)^{2\d-1}  \dparz u}+ \Re\innp{\nabla_\bot u } { \nabla _\bot \h_R^2 (1+ \abs x)^{2\d-1}  \dparz u}\\
& \quad + \Re\innp{ W_1(x) + \tilde W_3(x) - \z_W(x,\z)^2 } {\h_R^2 (1+ \abs x)^{2\d-1}  \dparz u}\\
& =: A_1(R) + A_2(R) + A_3(R) + A_4(R) + A_5(R).
\end{align*}

\noindent
{\bf 3.}
We have
\begin{align*}
 A_1(R) + A_5(R)
& = \Re\innp{\dparz u } { \dpar \h_R^2 (1+ \abs x)^{2\d-1}  \dparz u}\\
& \quad  + \Re \innp{-i(\partial_r \z_W(x,\z) + \z_W(x,\z) \dparz) u } { \h_R^2 (1+ \abs x)^{2\d-1}  \dparz u}.
\end{align*}
According to Lemma \ref{lem-est-udu} $\dparz u$ belongs to $L^{2,\d}(B_1^c)$ and hence
\begin{align*}
\liminf_{R \to \infty} \Re\innp{\dparz u } { \dpar \h_R^2 (1+ \abs x)^{2\d-1}  \dparz u}
&  = \frac 12 \liminf_{R \to \infty}\left( \innp{\dparz u }{ \partial_r(\h_R^2 (1+ \abs x)^{2\d-1} ) \dparz u}\right) \\
&\geq \left(\d- \frac 12\right) \nr{\h \dparz u}^2_{\nlpds{\d-1}} - c \nr{\dparz u}^2 _{L^2(B_2 \cap B_1^c)}. 
\end{align*}
Moreover
\begin{eqnarray*}
\lefteqn{ \liminf_{R \to \infty} \Re \innp{-i(\partial_r \z_W(x,\z) + \z_W(x,\z) \dparz) u } { \h_R^2 (1+ \abs x)^{2\d-1}  \dparz u}}\\
&\hspace{0.5cm}& \geq  \liminf_{R \to \infty} \Re \innp{-i\partial_r \z_W(x,\z)  u } { \h_R^2 (1+ \abs x)^{2\d-1} \dparz u}\\
&& \geq -c \nr u _{\nlpds{-1-\rho + \d}} \nr{\h \dparz u}_{\nlpds{\d-1}} .
\end{eqnarray*}
According to Proposition \ref{prop-reg-int}, these two estimates give
\begin{equation} \label{A1A5}
 \liminf_{R \to \infty} (A_1(R) + A_5(R)) \geq  \left(\d- \frac 12\right) \nr{\h \dparz u}_{\nlpds{\d-1}}  - c \Nc_d \Nc_u.
\end{equation}
With the same kind of argument and using the last property of Proposition \ref{prop-dpar} we prove that
\[
 \liminf_{R \to \infty} A_2(R) \geq \left(\frac 32 - \d\right) \nr {\nabla_\bot u }^2 _ {\nlpds{-\d}}-  c \Nc _u \Nc_d.
\]
Using \eqref{zeta2} and Lemma \ref{lem4.1} we show that
\[
\liminf_{R \to \infty} A_3(R) \geq -  c \Nc _u \Nc_d.
\]
We have the same estimate for $A_4(R)$ and 
\[
\forall R > 2, \quad  \Re \innp{f} {\h_R^2 (1+ \abs x)^{2\d}  \dparz u} \leq c \Nc _u \Nc_d,
\]
so we finally have
\[
\left(\frac 32 - \d\right) \nr {\nabla_\bot u }^2 _ {L^{2,\d-1}(B_2^c)} + \left(\d- \frac 12\right) \nr{\h \dparz u}^2_{L^{2,\d-1}(B_2^c)}  \leq c \Nc _u \Nc_d,
\]
which, together with Proposition \ref{prop-reg-int}, concludes the proof.
\end{proof}


As in the self-adjoint case we also use the following estimate:

\begin{proposition} \label{lem2.5bis}
There exists $C$ such that for $R > 0$, $\z \in K^*$, $f \in L^{2,\d}(\R^n)$ and $u = (H-\z^2)\inv f$ we have
\begin{equation*}
\nr{u}_{L^{2,-\d}( B_R^c)}^2 \leq C R^{1-2\d} \left( \nr u _{L^{2,-\d}(\R^n)}^2 + \nr f _{L^{2,\d}(\R^n)}^2 \right).
\end{equation*}
\end{proposition}

\begin{proof}
For all $r > 0$ we have
\begin{equation} \label{eq-Sr} 
\abs{(\dpar -i\z) u}_\ldgr^2 = \abs{\dpar u + \z_2 u}_\ldgr^2 + \z_1^2 \abs{u}_\ldgr^2 + 2\z_1 \innp{V_2 u}{u}_{B_r} + 4 \z_1^2 \z_2 \nr{u}^2_{B_r} + 2 \z_1 \Im \innp{f}{u}_{B_r}
\end{equation}
(see \cite[Prop.3.4]{saito} in the self-adjoint case). In particular 
\[
\z_1^2 \abs u _\Sphr ^2 \leq \abs{(\dpar -i\z)u}_\Sphr^2 + 2 \z_1 \nr f _{\nlpds \d} \nr u _{\nlpds{-\d}} + c \nr u ^2 _{L^{2,-\d}(\R^n)}.
\]
We multiply by $(1+r)^{-2\d}$ and integrate from $R$ to $+\infty$ to prove the proposition.
\end{proof}

With all these estimates and uniqueness of the outgoing solution for the equation $(H-\l^2)u = 0$ we can now conclude that the limiting absorption principle holds as in the self-adjoint case. 

The first step is to prove that if we have the estimates \eqref{estim-lap} on $K^*$, then we have the limiting absorption principle. We refer to \cite[Lemma 2.6]{saito} and recall briefly the idea. If $\z_m \to \l$ and $u_m = (H-\z_m^2)\inv f$, then a subsequence of $\seq u m$ converges to some $u$ in $L^{2}_{\loc}(\R^n)$ according to Proposition \ref{prop-reg-int}. We obtain convergence in $L^{2,-\d}(\R^n)$ according to the last estimate of \eqref{estim-lap}. Using again Proposition \ref{prop-reg-int} we also have convergence in $H^2_{\loc}(\R^n)$. The limit $u$ is necessarily an outgoing solution for the equation $(H-\l^2)u=f$. And since we have uniqueness for such a solution, we actually have convergence of the whole sequence.

To apply this result, we still have to check the first estimate of \eqref{estim-lap}. This is a contridaction argument: we assume that we can find sequences $\seq f m \in L^{2,\d}(\R^n)^\N$ and $\seq zm \in (K^*)^\N$ such that 
\[
z_m \limt m \infty \l \in K \cap \R_+^*, \quad \nr{(H-z_m)\inv f_m}_{L^{2,-\d}} = 1 \quad \text{and} \quad  \nr {f_m}_{L^{2,\d}(\R^n)} = \littleo m \infty(1).
\]
All the estimates of \eqref{estim-lap} hold for this sequence and hence $(H-z_m)\inv f_m$ converges to an outgoing solution for the equation $(H-\l)u=0$, which must be zero. This gives a contradiction (see the proof of Theorem 2.3 in \cite{saito}).\\

For the dissipative Schrödinger operator everything holds as in the self-adjoint case on the upper half-plane:

\begin{proposition} \label{prop-lap-diss}
If $V_2 \geq 0$ then the assumptions and hence the conclusions of Proposition \ref{prop-lap} hold for any $K$.
\end{proposition}

\begin{proof}
If $V_2 \geq 0$ then the operator $H$ with domain $H^2(\R^n)$ is maximal dissipative, so its resolvent $(H-z)\inv$ is at least well-defined for $\Im z >0$. Moreover, if $u \in H^2_{\loc}(\R^n) \cap L^{2,-\d}(\R^n)$ is an outgoing solution for the equation $(H-\l^2)u=0$ for some $\l \in K \cap \R_+^*$, then according to \eqref{eq-Sr} the solution $u$ vanishes on the support of $V_2$ and hence is an outgoing solution for the equation $(H_1-\l^2)u=0$, where $H_1 = -\D + V_1(x)$ is the self-adjoint part of $H$. We know from \cite[§3]{saito} that such a solution must be zero.
\end{proof}

In the non-selfadjoint case we do not have such a systematic result. In order to prove Theorem \ref{th-lap} we use the fact that for $h>0$ small enough and $z \in \C_{I,+}$ we already have existence and uniform estimates for the resolvent $(\hh-z)\inv$.

\begin{proof} [Proof of Theorem \ref{th-lap}]
 Let $\th_0 \geq 0$ be such that $V_2(x) + \th_0 \pppg x ^{-1 -\rho} \geq 0$ for all $x \in\R^n$. There exist a neighborhood $I$ of $E$, $h_0 >0$ and $C>0$ such that for all $h \in ]0,h_0]$, $z \in \C_{I,+}$ and $\th \in [0,\th_0]$ we have
\[
 \nr{\pppg x ^{-\d} (\hh^\th - z)\inv \pppg x^{-\d}}_{\Lc(L^2(\R^n))} \leq \frac C h,
\]
where $\hh^\th = \hh - h \th \pppg x ^{-1-\rho}$. Indeed, Theorem \ref{th-estim} gives such an estimate for any fixed $\tilde \th \in [0,\th_0]$, and the perturbation argument already used to prove \eqref{estim-hth} gives an estimate uniform for $\th$ in a neighborhood of $\tilde \th$. Now let $\th_1 \in [0,\th_0]$ and assume that $u=0$ is the unique outgoing solution for the equation $(\hh^\th -\l)u=0$ when $\th \in [\th_1,\th_0]$, $\l \in I$ and $h \in ]0,h_0]$. Let $h\in]0,h_0]$, $\l \in I$, $\th \in [\th_1 - C/2,\th_1]$ and let $u \in H^2_\loc(\R^n) \cap L^{2,-\d}(\R^n)$ be an outgoing solution for the equation $(\hh^\th - \l)u = 0$. We have 
\[
(\hh^{\th_1} - \l) u = h(\th_1-\th) \pppg x ^{-1-\rho} u \in L^{2,\d}(\R^n).
\]
According to Proposition \ref{prop-lap}, the outgoing solution for the equation $(\hh^{\th_1} - \l) u = f$ is given by the limiting absorption principle and hence we have
\[
\nr{u}_{L^{2,-\d}(\R^n)} \leq \frac C h \nr{h (\th_1 - \th) \pppg x ^{-1-\rho} u}_{L^{2,\d}(\R^n)} \leq \frac 12 \nr u _{L^{2,-\d}(\R^n)}.
\]
This proves that $u = 0$.
\end{proof}

\section{Semiclassical Measure}  \label{sec-mesure}

We study in this section the semiclassical measures for the outgoing solution of \eqref{helmholtz} when the source term $f_h$ concentrates on a bounded submanifold of $\R^n$, $V_1$ is of long range and $V_2$ is of short range. We adapt to this purpose the proof given in \cite{art_mesure} for the dissipative setting.\\

Let $\G$ be a (bounded) submanifold of dimension $d \in \Ii 0 {n-1}$ in $\R^n$. We consider $A \in C_0^\infty(\G)$, $S \in \Sc(\R^n)$ and define
\begin{equation} \label{source}
 f_h(x) = \int_\G  A(z) S\left( \frac {x-z} h \right) \, d\sgamma (z),
\end{equation}
where $\sgamma$ is the Lebesgue mesure on $\G$. We can check that $f_h$ is microlocalized on $N\G$ and $\nr{f_h}_{L^{2,\d}(\R^n)} = O(\sqrt h)$ for any $\d > \frac 12$. Let $E > 0$ be an energy which satisfies assumption \eqref{hyp-amfaible}. We assume that 
\begin{equation} \label{hyp-propag}
 \forall z \in \G, \quad V_1(z) < E.
\end{equation}

Let 
\[
 \negg = N\G \cap p\inv(\singl E) = \singl{(z,\x) \in \G \times \R^n \tqe \x \bot T_z \G  \text{ and } \abs \x^2 = E- V_1(z)}.
\]
Assumption \eqref{hyp-propag} ensures that $\negg$ is a submanifold of dimention $n-1$ in $\R^{2n}$. The Riemannian structure $g$ on $\negg$ is defined as follows. For $(z,\x) \in \negg$ and $(Z,\Xi),(\tilde Z,\tilde \Xi) \in T_{(z,\x)}\negg \subset \R^{2n}$ we set
\[
 g_{(z,\x)} \big((Z,\Xi),(\tilde Z,\tilde \Xi) \big) = \innp{Z}{\tilde Z}_{\R^n} + \innp{ \Xi_\bot}{\tilde \Xi_\bot}_{\R^n},
\]
where $\Xi_\bot,\tilde \Xi_\bot$ are the orthogonal projections of $\Xi,\tilde \Xi \in \R^n$ on $(T_z\G \oplus \R \x)^\bot = T_\x ( N_z\G \cap \negg)$. We denote by $\snegg$ the canonical measure on $\negg$ given by $g$, and assume that
\begin{equation} \label{hyp-Phi}
 \snegg \left( \singl{ (z,\x) \in \negg \tqe \exists  t > 0 , \vf^t (z,\x) \in \negg} \right) = 0.
\end{equation}

\begin{theorem} \label{th-mesure}
Let $f_h$ be given by \eqref{source} and $u_h$ be the outgoing solution for the Helmholtz equation \eqref{helmholtz}. Let assumptions \eqref{hyp-amfaible}, \eqref{hyp-propag} and \eqref{hyp-Phi} be fulfilled.
\begin{enumerate}[(i)]
\item There exists a non-negative Radon measure $\m$ on $\R^{2n}$ such that
\begin{equation}  \label{lim-opquu}
\forall q \in C_0^\infty(\R^{2n}),\quad  \innp{\Opw(q) u_h}{u_h} \limt h 0  \int_{\R^{2n}} q \, d\m.
\end{equation}
\item This measure is characterized  by the following three properties:
\begin{enumerate} [a.]
\item $\m$ is supported in $p\inv(\singl{E})$.
\item For any $\s \in ]0,1[$ there exists $R \geq 0$ such that $\m = 0$ in the incoming region $\zoneS_-(R,0,-\s)$.
\item $\m$ is solution of the Liouville equation
\begin{equation} \label{liouville}
\{p,\m\} + 2V_2 \m =\pi (2\pi)^{d-n}  \abs{A(z)}^2 \abs \x \inv \big|\hat S (\x)\big|^2   \snegg,
\end{equation}
\end{enumerate}
\item These three properties imply that for $q \in C_0^\infty(\R^{2n})$ we have
\begin{equation} \label{expr-mu}
\int_{\R^{2n}}\!\! q \, d\mu = \int_0^{+\infty} \!\!\! \int_{\negg} \pi (2\pi)^{d-n}  \abs{A(z)}^2 \abs \x \inv \big|\hat S (\x)\big|^2  q(\vf^t(z,\x)) e^{ - 2\int_0^t V_2(\bar x (s,z,\x))\,ds} \,d\snegg(z,\x) \,dt.
\end{equation}
\end{enumerate}
\end{theorem}

Note that as in \cite{bony09} we can let $E$ depend on $h$ : $E_h = E_0 + h E_1 +o(h) \in \C_+$, where $E_0>0$ satisfies assumption \eqref{hyp-amfaible} and $\Im E_h \geq 0$. Then $V_2$ has to be replaced by $V_2 + \Im E_1$ in \eqref{liouville} and \eqref{expr-mu}.\\

We recall the sketch of the proof, discuss differences with the dissipative case and refer to \cite{bony09,art_mesure} for details.
We first remark that the limit \eqref{lim-opquu} is zero when $q \in C_0^\infty(\R^{2n})$ is supported outside $p\inv(\singl E)$.
Then the idea is to replace the resolvent which defines $u_h$ by the integral over finite times of the propagator. More precisely, for $T \geq 0$ and $h \in ]0,1]$ we set
\[
 u_h^T = \frac ih \int_0^\infty \h_T(t)  e^{-\frac {it}h (\hh-E)} f_h,
\]
where $\h_T (t) = \h(t-T)$ and $\h \in C^\infty(\R,[0,1])$ is equal to 1 in a neighborhood of $]-\infty,0]$ and equal to 0 on $[\t_0,+\infty[$ for some $\t_0 > 0$ small enough (see \cite{art_mesure}). Then we can study separately the contribution of different times. For small times we proceed exactly as in the dissipative case. For intermediate times, and then to prove that $u_h^T$ is in some sense a good approximation of $u_h$ for large $T$ and small $h$, we need a non-selfadjoint version of Egorov's Theorem. \\

According to \eqref{res-non-diss} and Hille-Yosida's Theorem (see for instance Theorem II.3.5 in \cite{engel2}) we know that $\hh$ generates a continuous semi-group, which we denote by $\uh(t)$, and
\[
\forall t \geq 0, \quad \nr{\uh(t)}_{\Lc(L^2(\R^n))} \leq e ^{t m_-}
\]
(we recall that $m_- = -\inf V_2$).
 Let $ W_2,\tilde W_2 \in \symbor(\R^n,\R)$, $W = W_2 + \tilde W_2$ and, for $t \in\R$:
\[
\udh(t) = e^{-\frac {it}h (\huh -ihW_2)} \quad \text{and} \quad  \tudh(t) = e^{-\frac {it}h (\huh -ih \tilde W_2)} .
\]
The Egorov's Theorem extends without modification to the non-dissipative case:

\begin{theorem}\label{th-egorov}
Let $a \in \symbor(\R^{2n})$. There exist a family of symbols $\a_j(t)$ for $j\in\N$ and $t \in \R$ such that:
\begin{enumerate}[(i)]
\item  For all $t \in \R$, $N\in\N$ and $\hou$ we have
\begin{equation*} 
\udh(t)^* \Opw(a) \tudh(t) = \sum_{j=0}^N h^j \Opw(\a_j(t))  + h^{N+1} R_N(t,h),
\end{equation*}
where $R_N(t,h)$ is bounded on $L^2(\R^n)$ uniformly in $\hou$ and $t \in [0,T]$ for any $T \geq 0$.
\item We have
\[
\a_0(t)  = (a \circ \vf^t) e^{- \int_0^t W \circ \vf^s \, ds}.
\]
\item
For $t\in \R$ and $j\in\N$ we have 
\[
\supp \a_j (t) \subset \vf^{-t} (\supp a).
\]
\end{enumerate}
\end{theorem}


Let $w \in \R^{2n}$, $T \geq 0$ and $0 < t_{w,1} < \dots < t_{w,K^T_w} \leq T+ \t_0$ be the times between 0 and $T + \t_0$ for which $\vf^{-t_{w,k}}(w) \in \negg_0 = \negg \cap (\supp A \times \R^n)$. For $\t_w>0$ small enough we consider $\h_w \in C_0^\infty(]0,2\t_w[)$ equal to 1 in a neighborhood of $\t_w$. For $k \in \Ii 1 {K_w^T}$ we prove that in $L^2(\R^n)$
\[
 \frac ih \int_0^\infty \h_T(t) \h_w(t-t_{w,k} + \t_w) e^{-\frac {it}h (\hh-E)}  f_h = B_{w,k}^T(h) + \bigo h 0 \big(\sqrt h \big),
\]
where $B_{w,k}^T(h) $ is a lagrangian distribution of lagrangian submanifold 
\[ 
 \L_{w,k}^T= \{ \vf^t(z,\x), (z,\x) \in \negg, t \in ]t_{w,k}- \t_w,t_{w,k}+ \t_w[\}.
\]
This means that there exist $N \in \N$, $b_{w,k}^T \in C_0^\infty(\R^{n+N})$ and a non-degenerate phase function $\p \in \symbor(\R^{n+N},\R)$ (if $\nabla_\th \p(x,\th)=0$ for some $(x,\th) \in \R^{n+N}$, then the $N$ linear forms $d_{(x,\th)} \partial_{\th_i} \p : \R^{n+N} \to \R$, $1\leq i \leq N$, are linearly independant) such that
\[
 B_{w,k}^T(h) = \frac 1 {(2\pi h)^{\frac N 2} } \int_{\R^N} e^{\frac ih \p(x,\th)} b(x,\th)\, d\th ,
\]
and 
\[
 \singl{ (x,\nabla_x \p(x,\th)) \text{ for } (x,\th) \in \R^{n+N} \text{ such that } \nabla_\th \p (x,\th) = 0} \subset \L_{w,k}^T
\]
(this replaces what is said in \cite{art_mesure}).
This is proved by direct computations when $t_{w,k}$ is replaced by $\t_w$ (and $N=0$ in this case), and then we use the fact that for any $t > 0$ the propagator $\uh(t)$ can be seen as a Fourier Integral Operator and maps a lagrangian distribution of submanifold $\L$ to some lagrangian distribution of submanifold $\vf^t (\L)$. We know that for such a lagrangian distribution there exists a smooth and non-negative function $\nu_{w,k}^T$ on $\L_{w,k}^T$ such that
\[
 \forall q \in C_0^\infty(\R^{2n}), \quad \innp{\Opw(q)  B_{w,k}^T(h)}{ B_{w,k}^T(h)} \limt h 0 \int_ {\L_{w,k}^T} q(\tilde w) \nu(\tilde w) d\s_{\L_{w,k}^T}(\tilde w),
\]
where $\s_{\L_{w,k}^T}$ is the Lebesgue measure on ${\L_{w,k}^T}$. According to Egorov's Theorem, times far from 0 and $t_{w,k}^T$ ($k\in\Ii 1 {K_w^T}$) do not give any contribution around $w$ at the limit $h \to 0$, so we can prove that \eqref{lim-opquu} holds for some measure $\m_T$ if $u_h$ is replaced by $u_h^T$.\\

It remains to study the contribution of large times. In \cite[Prop. 2.3]{art_mesure} we used the fact that the damping factor $\exp\big(-\int_0^t V_2 \circ \vf^s \, ds\big)$ is a non-increasing function of $t$, which is no longer the case. We use Proposition \ref{prop-gros-amort-2} instead:

\begin{proposition}  \label{prop-super-egorov}
Let $J$ be a neighborhood of $E$ such that assumption \eqref{hyp-amfaible} holds for all $\l \in J$. Let $K_1$ and $K_2$ be compact subsets of $p\inv (J)$. Let $\e > 0$. Then there exists $T_0 \geq 0$ such that for $q_1,q_2 \in C_0^\infty(\R^{2n})$ respectively supported in $K_1$ and $K_2$ we have
\[
\forall T \geq T_0, \quad \limsup_{h \to 0} \nr{\Opw(q_1) U_h(T) \Opw(q_2)}_{\Lc(L^2(\R^n))} \leq \e \nr {q_1}_\infty \nr {q_2}_\infty.
\]
\end{proposition}


\begin{proof}
Let $q_1,q_2 \in C_0^\infty(\R^{2n})$ be respectively supported in $K_1$ and $K_2$. Let $t \mapsto \uuh(t)$ denote the unitary group generated by the self-adjoint part $\huh$ of $\hh$. According to Egorov's Theorem we have for all $T\geq 0$:
\begin{eqnarray*}
\lefteqn{\nr{\Opw(q_1) U_h(T) \Opw(q_2)}_{\Lc(L^2(\R^n))} = \nr{U_1^h(T)^* \Opw(q_1) U_h(T) \Opw(q_2)}_{\Lc(L^2(\R^n))}}\\
&\hspace{2cm}& = \nr{\Opw\left(q_2 (q_1\circ \vf^T)e^{-\int_{0}^T V_2 \circ \vf^s\, ds} \right)}_{\Lc(L^2(\R^n))} + \bigo h 0 (h)\\
&& \leq C \sup_{w\in\R^{2n}} \abs{q_2(w) q_1 \big( \vf^T(w)\big) e^{-\int_{0}^T V_2 ( \bar x (s,w)) \, ds}} + \bigo h 0 \big( \sqrt h \big),
\end{eqnarray*}
where the size of the rest depends on $T$, $q_1$ and $q_2$. The constant $C$ only depends on the dimension $n$.
According to Proposition \ref{prop-gros-amort-2} there exists $T_0$ such that for $T \geq T_0$ and $w \in K_2$ we have
\[
C e^{-\int_0^T V_2(\bar x (s,w))\,ds} \leq  \e  \quad \text{or} \quad \vf^T (w) \notin K_1.
\]
Therefore we have for all $T \geq T_0$:
\[
\begin{aligned}
\nr{\Opw(q_1) U_h(T) \Opw(q_2)}_{\Lc(L^2(\R^n))} \leq \e \nr {q_1}_\infty \nr {q_2}_\infty + \bigo h 0 \big( \sqrt h \big).
\end{aligned}
\]
It only remains to take the limit $h\to 0$ for fixed $T$, $q_1$ and $q_2$ to conclude.
\end{proof}

For the rest of the proof we proceed as in the dissipative case. 
We only have to be careful for the proof of Lemma 5.4 in \cite{art_mesure} since the resolvent $(\hh-z)\inv$ cannot be written as the integral of the propagator over positive times for all $z \in\C_+$. However, for $h>0$ small enough and $z \in\C_+$ close to $E$, $(\hh-z)\inv f_h$ is well-defined and belongs to $H^2(\R^n)$. Therefore we can write
\begin{eqnarray*} 
\lefteqn{\Opw(q) (\hh-z)\inv f_h - \Opw(q) e^{-\frac {iT} h (\hh-z)} (\hh-z)\inv f_h}\\
&& = - \int_0^T \Opw(q) \frac d {dt} e^{-\frac {it} h (\hh-z)} (\hh-z)\inv f_h\,dt\\
&& = \frac ih \int_0^\infty \h_T(t) \Opw(q)  e^{-\frac {it} h (\hh-z)} f_h \, dt - \frac ih \int_0^\infty \h(t) \Opw(q)  e^{-\frac {iT} h (\hh-z)} e^{-\frac {it} h (\hh-z)} f_h \,dt.
\end{eqnarray*}
Note also that since $\uh(T)$ is no longer estimated uniformly by 1, some rests depends on $T$ in the proof of this lemma. This is not a real problem since we take the limit $h \to 0$ for fixed $T$. We finally obtain that for any compact subset $K$ of $p\inv(J)$ and $\e > 0$ there exists $T_0$ such that for $q \in C_0^\infty(\R^{2n})$ supported in $K$ we have
\[
 \forall T \geq T_0, \quad \limsup_{h \to 0} \abs{\innp{\Opw(q) u_h}{u_h} - \innp{\Opw(q) u_h^T}{u_h^T}} \leq \e \nr {q}_\infty .
\]
With this estimate we can check that for all $q \in C_0^\infty(\R^{2n})$ the function $T \mapsto \int q \, d\m_T$ has a limit when $T \to +\infty$, that this limit defines a non-negative measure $\m$ on $\R^{2n}$, and finally that \eqref{lim-opquu} holds for this measure.
All the properties of $\m$ stated in Theorem \ref{th-mesure} are proved as in the dissipative case (in particular we use Theorem \ref{th-incoming} to prove (ii) b.).

\appendix

\section{Construction of an escape function} \label{sec-escape}

In this appendix we prove Proposition \ref{prop-fonc-fuite}. A similar result (with an inequality) is proved in \cite{jecko04}. The purpose was to give a proof which could be extended for matrix-valued operators. The version we give here is more convenient in our context.\\

Let $J = \left] \frac E 2 , 2E \right[$, $\s \in \left] 0 , \frac 12\right[$ and $\Rc$ given by Proposition \ref{prop-non-entr}. We set
\[
\Zc_{J,\pm} = \Zc_\pm \left(\Rc, 0, \mp \s \right)\cap p\inv(J).
\]

\begin{proposition}
If $\Rc$ is chosen large enough, then for $\a,\b \in\N^n$ such that $\abs \a + \abs \b \geq 1$ there exists $c_{\a,\b}$ such that for $t \geq 0$ and $(x,\x) \in \Zc_{J,\pm}$ we have
\begin{equation} \label{estim-derflot1}
\abs{\partial_x^\a  \partial_\x^\b \vf^{\pm t}(x,\x)} \leq c_{\a,\b} \pppg t.
\end{equation}
\end{proposition}

We know (see for instance Lemma IV.9 in \cite{robert}) that the derivatives of the flow $\vf^t$ are uniformly bounded as long as $t$ stays in a bounded subset of $\R$, but may grow exponentially fast with time. The purpose of this proposition is to check that if we only look at the flow far from the origin (where it is ``almost free'') then we recover a growth of size $O(t)$ as in the free case $(x,\x) \mapsto (x + 2t\x,\x)$.

\begin{proof}
{\bf 1.}
We prove the proposition for $(x,\x) \in \Zc_{J,+}$, the case $(x,\x) \in \Zc_{J,-}$ being analogous.
Let
\[
A(t,x,\x) = \begin{pmatrix} J_x \bar x (t,x,\x) & J_\x \bar x (t,x,\x)  \\ J_x \bar \x (t,x,\x)  & J_\x \bar \x (t,x,\x)  \end{pmatrix}\in M_{2n}(\R),
\]
where for instance $J_\x \bar x$ denotes the partial jacobian matrix of $\bar x$ with respect to $\x$. Suppose that
\[
\limsup_{t \to +\infty} \frac {\nr{A(t)}_{L^\infty(\Zc_{J,+} , M_{2n}(\R) )}} t = +\infty.
\]
Differentiating the system \eqref{syst-ham} with respect to $x$ and then to $\x$, we see that
\[
\partial_t A(t,x,\x) = B(t,x,\x) \cdot A(t,x,\x)
\]
where
\[
 B(t,x,\x) = \begin {pmatrix} 0 & 2 \Idm_n \\ - \Hess V_1 (\bar x (t,x,\x)) & 0 \end{pmatrix} \in M_{2n}(\R),
\]
and hence:
\[
\partial_t^2 A(t,x,\x) = \partial _t B(t,x,\x)\cdot A(t,x,\x) + B(t,x,\x)^2 \cdot A(t,x,\x) =: C(t,x,\x)  \cdot A(t,x,\x)
\]
According to Proposition \ref{prop-non-entr}, there exists $c_0 > 0$ such that for $(x,\x) \in \zoneS_{J,+}$ and $t\geq 0$ we have $\abs{\bar x (t,x,\x)} \geq c_0 (\abs x + t)$, so
\[
\nr{C(t,x,\x)} \leq \nr{\partial_t B(t,x,\x)} +\nr{ B(t,x,\x)^2} \leq c (\abs x +  t)^{-2-\rho},
\]
where $c$ depends neither on $(x,\x) \in \zoneS_{J,+}$, nor on $t \geq 0$. For $m\in\N$ let
\[
t_m = \inf \singl{ t\geq 1 \tqe \nr{A(t,x,\x)}_{L^\infty(\Zc_{J,+},M_{2n}(\R))} \geq mt}.
\]
Since the derivatives of $\vf^t$ are in $L^\infty(\R^{2n})$ uniformly for $t$ in a compact subset of $\R$, we have $t_m \to +\infty$. According to Taylor's formula we have
\[
\begin{aligned}
\frac {\nr{A(t_m)}_{L^\infty(\Zc_{J,+},M_{2n}(\R))}} {t_m}
& \leq \frac {\nr{A(0)}_{L^\infty(\Zc_{J,+},M_{2n}(\R))}}{t_m} + \nr{\partial_t A(0)}_{L^\infty(\Zc_{J,+},M_{2n}(\R))}\\
& \quad  + \frac 1 {t_m} \int_0^{t_m} (t_m-s) c ( \Rc + s)^{-2-\rho} \nr{A(s)}_{L^\infty(\Zc_{J,+},M_{2n}(\R))} \,ds\\
& \leq c + {c } \int_0^{t_m}\frac{t_m-s}{t_m}  (\Rc + s)^{-2-\rho}\, s \, m  \,ds\\
& \leq c + c m \int _0^{t_m} (\Rc+s)^{-1-\rho} \, ds\\
& \leq c + c m \Rc^{-\rho},
\end{aligned}
\]
where $c$ depends neither on $m \in \N$ nor on the choice of $\Rc$. If $ \Rc$ was chosen so large that $c \Rc^{-\rho} \leq \frac 14$, then the right-hand side is less than $m/2$ for large $m$, which gives a contradiction. The case $\abs \a + \abs \b = 1$ is proved.\\

\noindent
{\bf 2.}
We now proceed by induction on $\abs \a + \abs \b$. Let $\a,\b \in\N^n$ such that $\abs \a + \abs \b \geq 2$ and assume that the result is proved for any derivative of order less than $\abs \a + \abs \b$. For $j \in \Ii 1 n$, the differential operator $\partial_x^\a \partial_\x^\b$ applied to \eqref{syst-ham} gives
\[
\begin{cases}
\partial_t \partial_x^\a \partial_\x^\b \bar {x_j}(t,x,\x)  = 2 \partial_x^\a \partial_\x^\b \bar {\x_j}(t,x,\x) \\
\partial_t \partial_x^\a \partial_\x^\b \bar {\x_j}(t,x,\x)  = - \sum_{l=1}^n (\partial_{x_j}\partial_{x_l} V_1) (\bar x (t,x,\x)) \, \partial_x^\a \partial_\x^\b \bar {x_l} (t,x,\x) + b_{\a,\b,j}(t,x,\x),
\end{cases}
\]
where $b_{\a,\b,j}$ is a sum of terms of the form
\[
-(\partial_{x_j} \partial^\nu V_1) (\bar x (t,x,\x)) \prod_{k=1}^{\abs \nu} \partial_x^{\a_k} \partial_\x^{\b_k} \bar {x_{j_k}}(t,x,\x)
\]
where $\abs \nu \geq 2$, $\sum_{k=1}^{\abs \nu} \a_k = \a$, $\sum_{k=1}^{\abs \nu} \b_k = \b$ and for $k \in \Ii 1 {\abs \nu}$: $j_k \in\Ii 1 n$ and $\abs {\a_k} + \abs {\b_k} \geq 1$. In particular for all $k$ we have $\abs {\a_k} + \abs {\b_k} < \abs \a + \abs \b$, so each term is estimated by
\[
\abs{(\partial_{x_j} \partial^\nu V_1) (\bar x (t,x,\x)) \prod_{k=1}^{\abs \nu} \partial_x^{\a_k} \partial_\x^{\b_k} \bar {x_{j_k}}(t,x,\x)} \leq c (\abs x + t)^{-1-\rho -\abs \nu} \pppg  t ^{\nu} \leq c(\abs x + t)^{-1-\rho}
\]
where $c$ depends neither on $t\geq 0$ nor on $(x,\x) \in\zoneS_{J,+}$, and hence
\[
\abs{b_{\a,\b,j}(t,x,\x)} \leq c (\abs x + t)^{-1-\rho}.
\]
We also have
\[
\abs{\partial_t b_{\a,\b,j}(t,x,\x)} \leq c (\abs x +t)^{-1-\rho}.
\]
If we set
\[
A_{\a,\b}(t,x,\x) = \begin{pmatrix} \partial_x^\a \partial_\x^\b \bar {x_1}(t,x,\x) \\ \vdots \\ \partial_x^\a \partial_\x^\b \bar {x_n}(t,x,\x) \\ \partial_x^\a \partial_\x^\b \bar {\x_1}(t,x,\x) \\ \vdots \\ \partial_x^\a \partial_\x^\b \bar {\x_n}(t,x,\x)\end{pmatrix} \in \R^{2n}
\quad \text{and} \quad
D_{\a,\b}(t,x,\x) = \begin{pmatrix} 0\\ \vdots \\ 0 \\ b_{\a,\b,1}(t,x,\x) \\ \vdots \\ b_{\a,\b,n}(t,x,\x)\end{pmatrix} \in \R^{2n} ,
\]
we have
\[
\partial_t A_{\a,\b}(t,x,\x) = B(t,x,\x)\cdot  A_{\a,\b}(t,x,\x) + D_{\a,\b}(t,x,\x)
\]
and
\[
\partial^2_t A_{\a,\b}(t,x,\x) = C(t,x,\x)\cdot A_{\a,\b}(t,x,\x) + B(t,x,\x) \cdot D_{\a,\b}(t,x,\x) + \partial_t D_{\a,\b}(t,x,\x).
\]
Then we can conclude as above. Note that the matrix $C$ is the same, and hence the choice of $\Rc$ does not depend on $(\a,\b)$.
\end{proof}

The estimate of size $O(t)$ for $\bar x$ is what was expected since this is indeed what we have in the free case, but we can improve the result for $\bar  \x$:

\begin{corollary}
For any $\a,\b \in\N^n$ there exists $c_{\a,\b} \geq 0$ such that for $t \geq 0$ and $(x,\x) \in \zoneS_{J,\pm}$ we have
\[
\abs{\partial_x^\a \partial_\x^\b \bar \x (\pm t,x,\x)} \leq c_{\a,\b}.
\]
\end{corollary}

\begin{proof}
Let $\a,\b \in\N^n$. We have proved that
\begin{align*}
\partial_t \partial_x^\a \partial_\x^\b \bar {\x_j}(\pm t,x,\x)
&  = \mp  \sum_{l=1}^n (\partial_{x_j}\partial_{x_l} V_1) (\bar x (\pm t,x,\x)) \, \partial_x^\a \partial_\x^\b \bar {x_l} (\pm t,x,\x) + \bigo t {+\infty} (t^{-1-\rho}),
\end{align*}
where the rest is uniform in $(x,\x) \in \zoneS_{J,\pm}$. With the estimates we now have on the derivatives of $\bar x$ this means that there exists $c_{\a,\b} \geq 0$ such that for all $(x,\x) \in\zoneS_{J,\pm}$ and $t\geq 0$ we have
\begin{align*}
\abs{\partial_t \partial_x^\a \partial_\x^\b \bar {\x_j}(\pm t,x,\x)}
\leq c_{\a,\b} \pppg t^{-1-\rho}.
\end{align*}
It only remains to integrate in time to conclude.
\end{proof}

\begin{corollary}
For any $\a,\b \in\N^n$ there exists $c_{\a,\b} \geq 0$ such that for $t \geq 0$ and $(x,\x) \in \zoneS_{J,\pm}$ we have
\[
\abs{\partial_x^\a \partial_\x^\b \frac {\bar x (\pm t,x,\x)\cdot \bar \x (\pm t,x,\x)} {\abs {\bar x (\pm t,x,\x)} \abs{\bar \x (\pm t,x,\x)}}} \leq c_{\a,\b}.
\]
\end{corollary}

\begin{corollary}
Let $\d > \frac 12$. Then for any $\a,\b \in\N^n$ there exists $c_{\a,\b} \geq 0$ such that for all $t \geq 0$ and $(x,\x) \in \zoneS_{J,\pm}$ we have
\[
\abs{\partial_x^\a \partial_\x^\b \pppg {\bar x (\pm t,x,\x)}^{-2\d}} \leq c_{\a,\b} (\abs x +  t)^{-2\d} .
\]
\end{corollary}

\begin{proof}
Let $\a,\b\in\N^n$. We remark that $\partial_x^\a \partial_\x^\b \pppg {\bar x (\pm t,x,\x)}^{-2\d}$ is a sum of terms of the form
\[
c_K(\bar x (\pm t,x,\x)) \pppg{\bar x(\pm t,x,\x)}^{-2\d-K} \prod _{k=1}^K \partial_x^{\a_k}\partial_\x^{\b_k} \bar x (\pm t,x,\x)
\]
where $K \in \Ii 1 {\abs \a + \abs \b}$, $c_K(x) = \pppg x^{2\d + K} \frac {d^K}{dx^K} \pppg x^{-2\d}$ is bounded, $\a = \sum_{k=1}^K \a_k$ and $\b = \sum_{k=1}^K \b_k$.
\end{proof}

Now we can prove Proposition \ref{prop-fonc-fuite} :

\begin{proof}
Let $\tilde \h \in C_0^\infty(\R)$ be supported in $J$ and equal to 1 in a neighborhood of $E$.
Let $\h_+, \h_- \in C^\infty(\R)$ such that $\supp \h_+ \subset ]-\s,+\infty[$, $\supp \h_- \subset ]-\infty,\s[$ and $\h_+ + \h_- = 1$ on $\R$.
Let $\h \in C_0^\infty(\R^n,[0,1])$ be equal to 1 on ${B_{ \Rc+1}}$. Let
\[
g_\pm : (x,\x) \mapsto \h_\pm\left(\frac { x \cdot \x}{\abs x \abs \x} \right)  (1- \h(x)) \tilde \h \big(p(x,\x)\big) \pppg x^{-2\d}
\]
and, for $w \in \R^{2n}$:
\[
f_\pm(w) = \pm \int_0^{+\infty}  g_\pm (\vf^{\mp t} (w)) \,  dt .
\]
Let $w=(x,\x) \in \R^{2n}$. There exists $T_w \geq 0$ such that $\vf^\pm (\Zc_{J,\pm}) \cap B_x(2\abs x) = \emptyset$ for all $t \geq T_w$ and hence
\[
\forall v \in B_x(2\abs x), \forall t \geq T_w, \quad g_\pm \big( \vf^{\mp t}(v)\big) = 0.
\]
According to the regularity theorems under the integral sign the functions $f_+$ and $f_-$ are smooth around $w$. And hence on $\R^{2n}$. Moreover their derivatives along the flow $\vf^t$ are given by
\[
\{ p , f_\pm \} = \pm \int_0^{+\infty} \{ p , g_\pm \circ \vf^{\mp t}\} \, dt = g_\pm.
\]
We now check that all the derivatives of $f_\pm$ are bounded. For $\a,\b \in \N^n$ there exists $c_{\a,\b} \geq 0$ such that for $(x,\x) \in \Zc_{J,\pm}$ and $t\geq 0$ we have
\[
\abs { \partial_x^\a \partial_\x^\b (g_\pm \circ \vf^{\pm t}) (x,\x)} \leq c_{\a,\b} (\abs x + t)^{-2\d}.
\]
Let $w \in \R^{2n}$ such that $\vf^{\mp t}(w) \in  \Zc_{J,\pm}$ for some $t \geq 0$ (otherwise the derivatives of $f_\pm$ vanishes at $w$). Let $t_0$ denotes the maximum of such times $t$. We have
\begin{align*}
\abs{\partial_x^\a \partial_\x^\b f_\pm (w)}
& \leq  \int_0^{t_0} \abs{\partial_x^\a \partial_\x^\b \big(g_\pm \circ \vf^{\mp t} \big) (w) }\, dt = \int_0^{t_0}\abs{ \partial_x^\a \partial_\x^\b \big(g_\pm \circ \vf^{\pm t} \big)\big( \vf^{\mp t_0} (w)\big) } \, dt\\
&\leq c_{\a,\b} \int_0^{+\infty} (\Rc + t) ^{-2\d}\, dt.
\end{align*}
This means that $f_+$ and $f_-$ belong to $\symbor(\R^{2n},\R)$. It only remains to set $f = f_+ + f_-$ to conclude.
\end{proof}

\bibliographystyle{alpha}
\bibliography{bibliotex}

\vspace{1cm}

\begin{center}

\begin{minipage}{0.8 \linewidth}

\noindent  {\sc Laboratoire de mathématiques Jean Leray}

\noindent  {\sc 2, rue de la Houssinière - BP 92208}

\noindent  {\sc F-44322 Nantes Cédex 3}

\noindent  {\sc France}

 \begin{verbatim}
julien.royer@univ-nantes.fr
\end{verbatim}
\end{minipage}
\end{center}

\end{document}